\documentclass[sigconf,balance=false]{acmart}

\usepackage[utf8]{inputenc} 
\usepackage[T1]{fontenc}    
\usepackage{hyperref}       
\usepackage{url}            
\usepackage{booktabs}       
\usepackage{nicefrac}       
\usepackage{microtype}      
\usepackage{xcolor}         

\usepackage{amsmath,amsfonts,amsthm}
\usepackage{algorithmic}
\usepackage{graphicx}
\usepackage{textcomp}
\usepackage{xcolor}
\usepackage{algorithm}  
\usepackage{subcaption}
\usepackage{multirow}
\usepackage{diagbox}
\usepackage{indentfirst}
\usepackage{setspace}
\usepackage{stmaryrd}
\usepackage{enumitem}

\usepackage{longtable}

\theoremstyle{definition}
\newtheorem{definition}{Definition}[section]
\newtheorem{theorem}{Theorem}

\setcopyright{none}
\renewcommand\footnotetextcopyrightpermission[1]{}
\acmConference{}
\acmBooktitle{}

\definecolor{burntorange}{rgb}{0.8, 0.33, 0.0}

\newcommand{\makevect}[1]{\ensuremath{\boldsymbol{#1}}}
\renewcommand{\vec}[1]{\makevect{#1}}

\newcommand{\F}{\mathbb{F}}

\mathchardef\mhyphen="2D
\newcommand{\msf}[1]{\mathsf{#1}}

\newcommand{\bs}[1]{\boldsymbol{#1}}
\newcommand{\globalm}{\boldsymbol{\mathcal{W}}}
\newcommand{\interout}{\boldsymbol{H}}

\newcommand{\PartySet}[1]{\mathcal{P}_{#1}}
\newcommand{\pclients}{\PartySet{C}}
\newcommand{\pservers}{\PartySet{S}}
\newcommand{\cl}[1]{E_{#1}}

\newcommand{\ash}[1]{\llbracket #1\rrbracket} 

\newcommand{\Sim}{\mathcal{S}}
\newcommand{\Adv}{\mathcal{A}}

\newcommand{\Prot}[1]{\ensuremath{\Pi_{\mathsf{#1}}}}
\newcommand{\PFirstA}{\Prot{Global \mhyphen Shuffle}}
\newcommand{\PFirstB}{\Prot{Global \mhyphen BMF}}
\newcommand{\PSecond}{\Prot{Global \mhyphen Local \mhyphen BMF}}
\newcommand{\PMult}{\Prot{Mult}}
\newcommand{\PAct}{\Prot{Act}}
\newcommand{\PDiv}{\Prot{Div}}
\newcommand{\PDP}{\Prot{DP}}
\newcommand{\PGS}{\Prot{GS}}
\newcommand{\PNM}{\Prot{Norm}}

\newcommand{\SHU}{\Prot{Shuffle}}

\newcommand{\FWD}{\Prot{FWD}}
\newcommand{\BWD}{\Prot{BWD}}

\newcommand{\appref}[1]{Appendix~\ref{app:#1}}

\newcommand{\figref}[1]{Fig.~\ref{fig:#1}}

\newcommand{\secref}[1]{Section~\ref{sec:#1}}

\newcommand{\figlab}[1]{\label{fig:#1}}

\usepackage{tcolorbox}
\tcbuselibrary{skins,breakable}

\newtcolorbox{protobox}[2][]{%
  enhanced,
  title        = {#2},
  attach boxed title to top left={xshift=+3mm,yshift*=-3mm},
  breakable    = false,
  colback      = white, 
  colframe     = black!75,
  fonttitle    = \bfseries,
  colbacktitle = black!10!white,
  coltitle     = black,
  #1
}

\newenvironment{protofig}[3]{
  \begin{figure}[!h]
    \newcommand{\FigCaption}{#2}
    \newcommand{\FigLabel}{#3}
    \begin{protobox}{#1}
    }{%
    \end{protobox}
    \vspace{-1em}
    \caption{\FigCaption}
    \figlab{\FigLabel}
  \end{figure}
}

\begin{document}
\title{Secure and Privacy-Preserving Vertical Federated Learning}

\author{Shan Jin}
\affiliation{%
  \institution{Visa Research}
  \city{} 
  \state{} 
  \country{} 
}
\email{shajin@visa.com}

\author{Sai Rahul Rachuri}
\affiliation{%
  \institution{Visa Research}
  \city{}
  \state{}
  \country{}}
\email{srachuri@visa.com}

\author{Yizhen Wang}
\affiliation{%
  \institution{Visa Research}
  \city{}
  \state{}
  \country{}
}
\email{yizhewan@visa.com}

\author{Anderson C.A. Nascimento}
\affiliation{%
  \institution{Visa Research}
  \city{}
  \state{}
  \country{}}
\email{annascim@visa.com}

\author{Yiwei Cai}
\affiliation{%
 \institution{Visa Research}
 \city{}
 \state{}
 \country{}}
\email{yicai@visa.com}

\renewcommand{\shortauthors}{Jin et al.}

\begin{abstract}
We propose a novel end-to-end privacy-preserving framework, instantiated by three efficient protocols for different deployment scenarios, covering both input and output privacy, for the vertically split scenario in federated learning (FL), where features are split across clients and labels are not shared by all parties. We do so by distributing the role of the aggregator in FL into multiple servers and having them run secure multiparty computation (MPC) protocols to perform model and feature aggregation and apply differential privacy~(DP) to the final released model. While a naive solution would have the clients delegating the entirety of training to run in MPC between the servers, our optimized solution, which supports purely global and also global-local models updates with privacy-preserving, drastically reduces the amount of computation and communication performed using multiparty computation. The experimental results also show the effectiveness of our protocols.  
\end{abstract}
\keywords{Vertical federated learning, differential privacy, secure multiparty computation}

\maketitle
\footnotetext{
This is a preprint of a paper accepted for publication in
Proceedings on Privacy Enhancing Technologies (PoPETs), 2026.
The final version will be available via the official PoPETs website.
}
\section{Introduction}


Federated Learning~\cite{mcmahan2017fl} is a powerful machine learning paradigm over distributed data sources. It allows multiple data sources, often dubbed as the clients, to collaboratively learn an ML model that outperforms each client's training on its own data. 
Privacy of clients' sensitive and proprietary data has been a core motivation of FL since its genesis. 
Recent advances in privacy-preserving federated learning (PPFL) have shown mechanisms that provably enhance data privacy during both communication and computation~\cite{yin2021survey}. 
In particular, secure aggregation~\cite{Bonawitz2017practical} and multiparty computation~\cite{cramer2015secure} {guarantee} that no party can infer any information about the clients' model updates from intermediate values during model aggregation, whereas differential privacy-preserving mechanisms~\cite{Dwork2006differential,dwork2014algorithmic} ensure the amount of information that can be inferred from the aggregated model is limited by a provable guarantee. 
These techniques provably enhance the privacy of FL in the classical FL regime, in which data are split by rows across the clients and the clients' model updates are homogeneous in format.

However, vertical federated learning~(VFL)--a {challenging} FL setting that attracts {an} increasing amount of attention lately~\cite{Liu2024vertical}--{compels us to rethink} the design of PPFL. 
In VFL, the data, including the label, can be vertically split among clients. Each client may have a unique view of the {dataset} that consists of only a subset of the columns.  
For example, banks and credit card companies may want to collaborate against fraud, but typically bank account information and transaction history of the credit card are not {aligned} in the feature space. Taking advertisement as another example:~Normally, {an} advertising platform may know the {browsing} history and interests of the user, but only the shopping platform knows whether the click has converted to an actual purchase.
The ubiquity of VFL in real-world applications calls for a VFL-friendly PPFL mechanism.

Unfortunately, existing PPFL mechanisms~
\cite{Bonawitz2017practical,Fereidooni2021safe} do not apply to VFL.
In particular, with each client only holding a partial view of the data, no client can compute the model output or the model update on its own. The flow of backpropagation is disrupted. Moreover, the client and the server may no longer have the same model architecture, i.e., the global model may not be a simple aggregation, e.g., averaging, of the {clients'} models. Instead, 
the entire model is often split into separate models between clients and the server. In a general VFL framework~\cite{Liu2024vertical}, a client model extracts {an} embedding from the client's private data as the intermediate output; the global model at the server then {proceeds} with the intermediate outputs from all clients to complete the learning task. 
Although delegating all computation to a secure MPC environment plus applying DP noise during training can ensure both input and output privacy, such a naive combination will be computationally prohibitive, especially when training deep neural {networks}. Previous work in~\cite{pentyala2022training} applied the same idea but restricted it {to} training a generalized linear model only. 
Besides, the work in~\cite{qiu2023efficient,qiu2024secure} designed a framework for training VFL securely based on secure aggregation~\cite{Bonawitz2017practical}. However, this framework did not provide any DP protection on each client's private data, which could easily leak information under membership inference attacks~(MIAs)~\cite{shokri2017membership,nasr2019inference,ye2022enhanced}. Most recently, a federated transformer framework for fuzzy VFL {was} proposed in~\cite{wu2024federated}. This work focuses on a special variant of conventional VFL, fuzzy VFL, which enables \emph{one-to-many} record linkage across clients. 
While this framework preserves both input and output privacy, it represents only a restricted case in VFL. In contrast, our work aims to develop a privacy-preserving protocol for the more general VFL setting. 

In response to {these challenges}, we propose a novel practical privacy-preserving vertical federated learning mechanism with privacy guarantees on both input and output privacy. 
Our main contributions are {as follows}:
\begin{enumerate}
    \item 
    We show an efficient {differentially} private VFL mechanism that couples secure aggregation with {the} Gaussian mechanism and {the} Matrix Factorization~(MF)-based mechanism to train a general \emph{global} model given pre-trained, frozen local models at the clients. (Section~\ref{sec:mpc_dp_localfreez}.)  
    \item 
    We extend the VFL mechanism to allow practical privacy-preserving updating of local layers within \emph{local} models as needed without adding DP noises to multiple local layers. To achieve {this}, we leverage linear estimation techniques to extract insights from the privacy-preserving information published by the server, enabling backward propagation on local models, under {the MF-based} mechanism. (Section~\ref{sec:mpc_dp_gl_update}.)

    
    \item We empirically evaluate both of our mechanisms on CIFAR-10/EMNIST. 
    Our {mechanisms show} competitive utility on both datasets across a wide range of privacy parameter $\varepsilon$s. (Section~\ref{sec:experiments}.)
\end{enumerate}
Our mechanism supports the general VFL pipeline in~\cite{Liu2024vertical}, which allows {a} flexible split of {the} global and local {models}. On a high level, our mechanism securely computes the global model updates and releases the global model with DP protection. This design makes the computation of the global model a choke-point for privacy, and as a result, {computation involving only local models at the clients} can be done in plaintext, which significantly reduces the computation overhead compared to running MPC end-to-end. This improvement enables training of complex architectures like ResNet-18, which was previously infeasible with end-to-end MPC. Furthermore, we leverage linear estimation techniques to recover the desired per-sample gradients from the received DP-protected release of batch gradient. By the post-processing property of DP, the estimated per-sample gradients are also DP protected, which can then be used by the clients {in the clear to perform backpropagation on} local models. 

\section{Related Work}

Various previous works have acknowledged and attempted to address the {challenge} of preserving privacy in VFL. On one hand, multiple encryption-based methods~\cite{fu22blind,huang2023efmvfl,mugunthan2021multivfl,qiu2023efficient,qiu2024secure} have been proposed to preserve input privacy in computation. However, this line of work does not prevent output privacy leakage, such as membership inference {attacks} against the released model. On the other hand, several existing mechanisms~\cite{vepakomma20202nopeek,wang2020hybrid} attempt to enhance output privacy by perturbing the intermediate results. However, these mechanisms either target relatively simple model architectures or lack rigorous privacy guarantees. {In~\cite{xie2024admm}, the authors present a differentially private VFL framework with multiple server heads, allowing clients to conduct multiple local updates per communication round through an ADMM~(Alternating Direction Method of Multipliers)-based optimization scheme. This approach achieves remarkable communication efficiency and also a better privacy–utility tradeoff compared to baseline methods such as split learning~\cite{vepakomma2018splitlearning}. However,  
ADMM’s sophisticated coordination mechanism requires extensive hyperparameter tuning, making it highly sensitive to practical deployment settings. For example, even a slight change in the penalty parameter $\rho$ can lead to significant performance variations. More importantly, the subsampling mechanism applied in this approach still requires synchronized randomness across all clients to maintain vertical data alignment, which breaks the anonymity of the sampling process and may lead to a privacy leakage vector. Recently, \cite{gai2025differentially} proposes a differentially private VFL framework that introduces adaptive constraints and dynamic noise injection to better balance the privacy–utility tradeoff. However, the framework still incurs performance degradation on a complex dataset~(CIFAR-10) under stringent privacy budgets, and the added adaptive mechanisms including per-iteration Laplacian score calculation also increase tuning overhead. In contrast, our work systematically studies the possibility of combining MPC and provable DP to enjoy the benefit of both worlds for achieving a better utility-privacy trade-off. }

In terms of combining MPC and DP for VFL, the most notable and the closest to ours in spirit is~\cite{wu2024federated}, which enhances both input and output privacy in fuzzy-linked VFL by applying local differential privacy~(LDP) to the local intermediate representations and aggregating them securely via MPC. This one-to-many fuzzy linkage enables subsampling for privacy amplification, but is not feasible in conventional one-to-one linkage VFL.  
Compared to~\cite{wu2024federated} which {is particularly suited to} VFL with {fuzzy-link} tasks, we propose a framework for {the} general VFL setting with careful consideration {of} all surfaces of privacy leakage including model {release}, intermediate output exchange, and data {alignment}. 

Additional references not directly related to our work are provided in~\appref{related_work}, for readers seeking additional background.

\section{Preliminaries}




\paragraph{Differential Privacy}



{The notion of DP was first introduced in~\cite{dwork2006diff}.  
Intuitively, differential privacy ensures that the participation of any individual has limited influence on the released output, thereby preventing adversaries, regardless of their auxiliary information, from reliably inferring whether a particular individual’s data was included in the dataset. 
In addition, DP satisfies a desirable post-processing property: The privacy guarantees are preserved, no matter how the output is later used or manipulated~\cite{dwork2014algorithmic}.
DP is particularly pertinent to ML because ML models, which are the outputs of learning algorithms, will be released or heavily queried at usage. Existing attacks already show extraction of input information, i.e., the training data, from unguarded models~\cite{shokri2017membership,nasr2019inference,ye2022enhanced}.
In this paper, we leverage multiple differentially private mechanisms, including the Gaussian mechanism via the moments accountant proposed in~\cite{abadi2016privacy}, as well as the Matrix Factorization mechanism~\cite{ChoquetteChoo2023multi,Kalinin2024banded}. The basic concepts and formal definitions of these mechanisms are summarized in~\appref{dp}.}


{\paragraph{Multiparty Computation} 
MPC allows mutually distrusting parties to jointly compute a function on their private data without revealing {the data}. It has found wide application, including PPML, and provides two core guarantees: \emph{correctness} (the output is correct) and \emph{privacy} (no information beyond the output is learned). Our work assumes a set of servers that can run MPC with flexible configurations (e.g., number of servers, corruption threshold). 
We focus on passive security, but by replacing the underlying components with actively secure protocols, one could achieve active security.
We leverage MPC protocols to 1) securely compute the model updates without revealing to a centralized server, and 2) re-enable classical DP mechanisms in a decentralized setting in the absence of a trusted curator.
In our protocols, data is secret-shared among servers using a linear scheme (e.g., arithmetic or replicated secret-sharing), enabling local additions and secure multiplications with minimal communication. We use the notation $\ash{\cdot}$ to indicate a linear {secret-sharing} scheme. We also use recent advances allowing mixed-mode computation, which {include} specialized protocols for non-linear activation functions, secure shuffles, and Gaussian noise sampling. For further details on MPC and the secret-sharing scheme, see~\appref{mpc}.}

\section{Problem Setup}
\label{sec:setup}

\subsection{Notation}


Typically, in a FL system, there is a central server (sometimes multiple servers) {that} seeks to train a model based on private data from $N$ clients. 
Each client $E_i$ has its own private dataset $\mathcal{D}_i$ and we use ${\mathcal{D} = \mathcal{D}_1 \bigcup  \mathcal{D}_2 \dots \bigcup \mathcal{D}_N}$ to denote the global dataset space. ${\mathcal{D}}$ contains the feature space ${\mathcal{X}}$, the label space ${\mathcal{Y}}${,} which contains $S$ classes, and the sample ID space ${\mathcal{I}}$. Hence the dataset owned by the $i$-th client is denoted by~$\mathcal{D}_i \triangleq: (\mathcal{X}_i, \mathcal{Y}_i, \mathcal{I}_i)$.

In this paper, we consider the particular scenario of vertical federated learning, where the clients hold a disparate set of features that come from the same sample ID space~(henceforth, we omit $\mathcal{I}$). More specifically, the feature space is divided into a finite number of sub-spaces with $\mathcal{X} = \mathcal{X}_1 \bigcup \mathcal{X}_2 \dots \bigcup \mathcal{X}_N$ where each client, $\cl{i}$, holds the sub-space $\mathcal{X}_i$.
We assume there are $M$ samples in the global dataset, $\mathcal{D} \triangleq: \{\boldsymbol{x}_m, {y}_m \}^M_{m = 1}$ where 
$\vec{x}_m \in \mathbb{R}^{d_{\vec{x}}}$ is {the feature vector of each sample, with} feature dimension $d_{\vec{x}}$. Note the feature dimension $d_{\vec{x}}$ typically refers to the size after flattening and does not affect the input shape fed into the model. Now {each} client $E_i$ only holds a subset of the features in $\vec{x}_m$, with feature dimension $d_{\vec{x}^{(i)}}$, such that $d_{\vec{x}} = \sum_{i = 1}^N d_{\vec{x}^{(i)}}$. Therefore, we denote $\{\vec{x}^{(i)}_{m} \in \mathbb{R}^{d_{\vec{x}^{(i)}}}\}^M_{m=1}$ as the samples held by the $i$-th client where $\boldsymbol{x}^{(i)}_{m}  \subseteq \mathcal{X}_i$ with $m \in [1, M]$
and hence $\vec{x}_m = \vec{x}^{(1)}_{m} \bigcup\vec{x}^{(2)}_{m}\dots \bigcup\vec{x}^{(N)}_{m}$.


In VFL we distinguish between two kinds of clients. One is called the \emph{label client}~(or \emph{active client}), which holds the samples with a subset {of} features as well as the labels $\{y_m\}^M_{m = 1}$. The other clients, which we call \emph{feature clients}~(or \emph{passive clients}), only hold a subset of the features for their corresponding samples. For simplicity, we assume there is only one label party, which is the $N$-th client $E_N$.

\subsection{System Model} We denote $\boldsymbol{\Theta}$ as the full model in VFL. We can decompose $\boldsymbol{\Theta}$ into a global model $\globalm$ parameterized by $\boldsymbol{\theta}$, and each client $E_i$ holds a local model $\boldsymbol{\mathcal{F}}_i$ parameterized by $\boldsymbol{\phi}_i$. The last client $E_N$ holds the labels $\vec y$, in addition. The loss function for $j$-th sample is defined as
\begin{align}
    \min_{\boldsymbol{\Theta}}{\mathcal{L}_j(\boldsymbol{\Theta}; \boldsymbol{x}_{j}, \boldsymbol{y}_{j})} = l\left(\globalm(\boldsymbol{\theta}; \{\boldsymbol{\mathcal{F}}_i(\boldsymbol{\phi}_i; \boldsymbol{x}^{(i)}_{j})\}_{i = 1}^N ), {y}_j \right)~~, \nonumber\label{eq:1}
\end{align}
where $l(\cdot)$ denotes the task loss. Since typically the training dataset is divided at {the} batch level, we set $B$ as the number of samples in a mini-batch and we assume $B_{\emph{num}}: = M/B $ batches are used in each epoch~(for simplicity, here we assume $M$ is divisible by $B$). We also denote the total number of epochs as~$E_\emph{num} $, and therefore the total number of steps $T_\emph{num} $ is obtained through $T_\emph{num}= B_\emph{num} \cdot E_\emph{num} $. 

Typically, during model training, at the training step $t$, the index of the batch used in this step, is denoted by $b$ where $b:= t\mod{B_\emph{num}} $. Hence we define the samples of $b$-th batch at $i$-th client as $\textbf{x}^{(i)}_b = \{\boldsymbol{x}^{(i)}_{b, j}\}^{B}_{j=1}$ with $b \in [1, M/B ]$, where $\boldsymbol{x}^{(i)}_{b, j}$ is the $j$-th sample of $b$-th batch at $i$-th client. Similarly, $\textbf{y}_b = \{{y}_{b, j}\}^{B}_{j=1}$ are the labels for the $b$-th batch. 

Then at the beginning of each step, each client $E_i$ feeds its own batch into the local model $\boldsymbol{\phi}_i$ and obtains the intermediate output: $\textbf{H}^{(i)}_b = \{ \interout^{(i)}_{b, 1};\dots;\interout^{(i)}_{b, B}\}\in \mathbb{R}^{ B \times d_{\interout^{(i)}}}$ where $\interout^{(i)}_{b, j} = \boldsymbol{\mathcal{F}}_i(\boldsymbol{\phi}_i; \boldsymbol{x}^{(i)}_{b, j}) \in \mathbb{R}^{d_{\interout^{(i)}}}$ with $j \in [1, B]$ and $d_{\interout^{(i)}}$ is the feature dimension of the intermediate output at $i$-th client. Then each $E_i$ uploads its intermediate output $\textbf{H}^{(i)}_b $ to the server, where the server concatenates {the} intermediate outputs from all clients to construct $\textbf{H}_b = \{\textbf{H}^{(1)}_b, \dots, \textbf{H}^{(N)}_b \} \in \mathbb{R}^{B \times d_{\textbf{H}} }$ where $d_{\textbf{H}} = \sum^N_{i = 1}d_{\interout^{(i)}}$. Eventually, {the} server feeds $\textbf{H}_b$ into the global model $\boldsymbol{\theta}$ and obtains the logits, $\boldsymbol{Z} \in \mathbb{R}^{B \times S}$, as the final outputs. For ease of reference, the notations introduced above are summarized in Table~\ref{tab:notation} in~\appref{add}.

\subsection{Security and Privacy Definitions}
\label{sec:security_definition}

\subsubsection{Data Alignment} We assume that the data instances (features and labels) among the clients are aligned:~data linked from the same sample ID space, also called \emph{one-to-one linkage}. 
Since we are in the cross-silo setting, 
for privacy protection purpose, 
the alignment can be performed through a privacy-preserving data alignment protocol such as Private Set Intersection~(PSI)~\cite{pinkas2018psi,CCS:KMPRT17}. 
This is inherent to any vertically split FL protocol, and we see the data alignment phase as orthogonal to the problem addressed in this paper.

{Note that we assume all clients are semi-honest, as defined in Section~\ref{threat_model}. Under this threat model, both entity alignment and data participation are strictly determined by the protocol specification. }

{\subsubsection{Threat Model}\label{threat_model}

\paragraph{Trust Assumptions}~We consider $K$ servers $\pservers = \{P_k\}^K_{k=1}$ (instantiated with $K\!=\!3$ in our implementation) connected via pairwise authenticated channels. We follow the three-party setting as it strikes a practical balance: it provides separation of trust among independently managed servers while keeping communication costs low, and it has been adopted in real deployments~\cite{CCS:AFLNO16}. A \emph{static}, \emph{semi-honest} adversary~$\Adv$ corrupts at most $t < K/2$ servers before protocol execution begins (i.e., $t\!=\!1$ when $K\!=\!3$);~the corrupted server(s) follows the protocol honestly but $\Adv$ observes its complete view (received messages from other parties and clients, randomness, and state). We adopt the semi-honest model (passive) because our cross-silo deployment targets servers operated by distinct organizations, where the reputational and legal cost of active deviation outweighs potential gains~\cite{matthew2025covert}. However, our framework is not tied to a specific MPC protocol. Since all our protocols only make \emph{black-box} use of standard MPC building blocks such as multiplication, secure comparison, etc., replacing them with maliciously (actively) secure counterparts in the UC-secure model, yields active security.

We assume all $N$ clients are semi-honest:~They correctly follow the defined protocols but may attempt to infer other clients’ and servers’ private information, including data features and labels, from the information exchanged between clients and servers during the protocol execution. We further consider external adversaries that can observe the prediction outputs and  gradients released during training, and query the deployed model after training, with the goal of inferring clients’ private information through membership inference attacks~\cite{shokri2017membership}. 

\paragraph{Input Privacy} Our framework employs MPC to provide input privacy. Each client secret-shares its data among multiple servers, and our protocols guarantee that \emph{no} intermediate values are revealed to any individual server under the semi-honest, honest-majority assumption~(Theorem~\ref{thm:input_privacy_main}).

\paragraph{Output Privacy} In our framework, we say the system provides $(\epsilon, \delta)$-output privacy if \emph{all} information released by the servers is $(\epsilon, \delta)$ differentially private~(\ref{def:dp}), with respect to each client $E_i$'s dataset $\mathcal{D}_i$, for all $i \in [1, N]$~(Section~\ref{subsec:outputprivacy}). 
}


\section{Single-Server Plaintext Model Training}
\label{sec:plian}


Before detailing our proposed protocols, we first introduce the basic VFL framework under a single server~\cite{Liu2024vertical} without any privacy protection, i.e. the plaintext model. This introduction will familiarize the readers with {the training flow of} a typical VFL framework while also {illustrating} the privacy challenges. 

\paragraph{Global Model Update}~{In the forward pass}, each client computes its local representation $\textbf{H}^{(i)}_b$ and {passes} to the server. The server concatenates the representations to $\textbf{H}_b$ and computes the logits $\boldsymbol{Z} = \globalm(\boldsymbol{\theta}; \textbf{H}_b)$ and subsequently the loss $\mathcal{L}$.
During back-propagation, the server computes the per-sample gradients for the global model, here the gradient of the $j$-th sample is represented by
\begin{align}
\boldsymbol{g}_{\boldsymbol{\theta}, j} = \frac{\partial \mathcal{L}_j}{\partial \boldsymbol{\theta}} = \frac{\partial l(\sigma(\boldsymbol{Z}_j),y_{b,j})}{\partial \boldsymbol{\theta}}~~, 
\end{align}
where $\sigma$ is the activation function and $\boldsymbol{Z}_j$ are the logits produced by the $j$-th sample. The aggregated per-batch gradient is~$\boldsymbol{g}_{\boldsymbol{\theta}} = \frac{1}{B}\sum^B_{j=1} \boldsymbol{g}_{\boldsymbol{\theta}, j}$.
Once the gradient is obtained, the server can update the global model using first-order optimization algorithms. For example, in the classic stochastic gradient descent~(SGD), a new global model $\boldsymbol{\theta}^{t+1}$ can be obtained as $\boldsymbol{\theta}^{t+1} = \boldsymbol{\theta}^{t} - \eta_s \boldsymbol{g}_{\boldsymbol{\theta}}$, where $\eta_s$ is the learning rate and $t$ is the training step. 

\paragraph{Local Models Update}~Updating the local model is slightly more involved than {updating} the global model but still achievable via collaboration between the server and the clients. By chain-rule, the per-sample gradient of the local model can be written as
\begin{align}
 \boldsymbol{g}_{\boldsymbol{\phi}_{i}, j} = \frac{\partial  \mathcal{L}_j}{\partial \boldsymbol{\phi}_i}= \frac{\partial {l}(\sigma(\boldsymbol{Z}_j), y_{b,j})}{\partial  \interout^{(i)}_{b,j}} \cdot \frac{\partial  \interout^{(i)}_{b,j}}{\partial \boldsymbol{\phi}_i}
\end{align}
The first term, $\partial l(\sigma(\boldsymbol{Z}_j), y_{b,j})/\partial  \interout^{(i)}_{b,j}$, is the gradient of the loss~w.r.t. (with respect to) the client's output at $j$-th sample. Note {that} the server can compute the per-sample gradients for each client's intermediate output and then {broadcast} the results back to the {respective} clients. The second term, $\partial \interout^{(i)}_{b,j}/\partial \boldsymbol{\phi}_i$, can be computed locally by the client. Then, each client can combine the two terms to recover the per-sample gradient of its local model and {update} the model using SGD through $\boldsymbol{\phi}_{i}^{t+1} = \boldsymbol{\phi}_{i}^t - \eta_i \boldsymbol{g}_{\boldsymbol{\phi}_{i}}$, where $\boldsymbol{g}_{\boldsymbol{\phi}_{i}} = \frac{1}{B}\sum^B_{j = 1} \boldsymbol{g}_{\boldsymbol{\phi}_{i}, j}$ and $\eta_i$ is the learning rate at the $i$-th client. 

The plaintext training protocol protects neither the input nor the output privacy. On the one hand, the intermediate outputs $\textbf{H}_b$ uploaded from the clients are completely disclosed to the server. On the other hand, the trained global model and the per-sample gradient with respect to the intermediate output~($\partial \mathcal{L}_j/\partial  \interout^{(i)}_{b,j}$) are disclosed without any protection:~Such disclosure will leak information {about} each client's private data to others as studied in previous research~\cite{zhu2019leakage,jin2021cafe,fu2022label}. These privacy challenges motivate us to design practical {secure} and privacy-preserving protocols for VFL.

\section{Our 
:~MPC-aided Vertical FL Training with DP}
\label{sec:mpc_dp}

In this section, we introduce our VFL training protocols that {guarantee} both input and output privacy. We consider a two-stage training scheme: a client-held section~(local model) and a server-held section~(global model) maintained in secret-shared form.
We use the global model as a choke-point for privacy. By applying MPC on the global model with DP protection instead of applying MPC end-to-end, our design significantly reduces MPC usage, which substantially improves the training time compared to na\"ive MPC \emph{end-to-end} approaches. Also, this improvement enables training of complex architectures, which was previously infeasible with end-to-end MPC.
In Section~\ref{sec:mpc_dp_localfreez}, we introduce our protocol with frozen local models. In Section~\ref{sec:mpc_dp_gl_update}, we show an extension that allows the clients to update their local models with privacy {guarantees} too.
Before proceeding to our protocols, we give a primer {on} the MPC sub-protocols and blocks used to build our protocols.


\paragraph{Sub-protocols for MPC}~
Specifically, we use the protocol from \cite{CCS:AFLNO16} as the basis to instantiate the 3PC building blocks of secret-sharing, reconstruction, and multiplication~$\PMult$. To instantiate the shuffling $\SHU$, we use the protocol from \cite{CCS:AHIKNPTT22}. We use the protocols from \cite{ACNS:AlySma19} for mathematical functions over fixed-point numbers such as approximated square root, exponentiation, and division, using which we build the $\ell_2$-normalization protocol $\PNM$ and the division protocol $\PDiv$. For more details about how these building blocks are implemented in MP-SPDZ, we refer the reader to \cite{CCS:Keller20}. Moreover, we build the forward-passing module $\FWD$ and the backward-propagation module $\BWD$ in MP-SPDZ, which are the wrap-ups of the sub-protocols including $\PMult$ and activation function $\PAct$ {introduced in}~\cite{payman2018aby3}, for conducting the forward-passing and backward-propagation.

The Gaussian noise protocol, \PGS, could be implemented from a number of existing protocols such as~\cite{keller2024secure,FC:EIKN21,CCS:ChasheUll19,das2025communication}. Given input $(T_\emph{num}, d, \sigma)$, where $d$ is the length of the noise vector, and $\sigma$ is the standard deviation of the {noise}, a secret-shared version of noise matrix $\mathcal{N}_{\text{Tab}} \in \mathbb{R}^{T_{num} \times d}$, {where} each noise vector follows the normal distribution of $\mathcal{N}(0, \sigma^2\boldsymbol{I}_d)$, 
is obtained among the servers. Here $I$ is an identity matrix. As the noise is not input data-dependent, we assume the servers compute the noise matrix in the pre-processing phase, before the start of the training phase.

\subsection{Training a Global model with DP}
\label{sec:mpc_dp_localfreez}

{We first consider the scenario in which the clients have pre-trained local models that will be frozen during the training phase; only the global model, $\globalm$, will be updated. Typically, to achieve output privacy guarantees, the servers can apply the DP-SGD algorithm~\cite{abadi2016privacy} during global model training. In practice, privacy amplification techniques, such as subsampling~\cite{abadi2016privacy}, are often employed to improve the privacy–utility tradeoff. These techniques require each mini-batch to be sampled uniformly at random from the training dataset, with the sampling permutation kept hidden from all parties; that is, no party has knowledge of the permutation.

However, this requirement poses a challenge in vertical federated learning:~As data holders, once data alignment is completed, all clients share knowledge of the dataset ordering. Consequently, under the threat model defined in~Section~\ref{threat_model}, mini-batch permutations cannot be privately generated and distributed to clients, neither directly by the servers nor through a joint agreement among clients that is subsequently revealed, without compromising the anonymity of sample participation.

Therefore, we present two variants of our protocol in order to optimize the privacy-utility trade-off under various privacy and computation constraints. 
Our first protocol utilizes the Gaussian mechanism from~\cite{abadi2016privacy} together with privacy amplification via subsampling~\cite{wang2019subsampled}. We use a secure shuffling procedure to enable subsampling without replacement and therefore name the protocol $\PFirstA$. 
Our second protocol targets the case where private shuffling of the data is computationally prohibitive. By utilizing the Banded Matrix Factorization approach from~\cite{Kalinin2024banded}, our second protocol, $\PFirstB$, can achieve comparable privacy-utility trade-off without relying on amplification by subsampling.

\subsubsection{Protocol with Amplification by Subsampling}
\label{sec:mpc_dp_localfreez_shuffle}
As mentioned above, to ensure anonymity of sample participation, if one intends to apply subsampling, the permutations cannot be directly provided to clients, neither given by the servers nor through an agreement among the clients themselves. Towards this, we propose an efficient protocol to achieve privacy amplification through subsampling that is also compatible with the VFL framework. Our full protocol consists of three main components: 1) forward-passing and secret-sharing on intermediate outputs, 2) secure computation of gradients for the global model, 3) privacy-preserving noise addition for the securely computed global model gradients and model update. }

\paragraph{Forward-passing and Secret-sharing}~During local forward passing, each client $E_i$ feeds all $M/B$ batches of its local data into the local model $\boldsymbol{\mathcal{F}}_i$, and obtains the intermediate outputs $\textbf{H}^{(i)}_b = \{ \interout^{(i)}_{b, 1};\dots;\interout^{(i)}_{b, B}\} $ for $b \in [1, M/B]$.
Notice that each local model is a \emph{deterministic function} since {it is} frozen during training.  
Furthermore, the order of the samples in each client is fixed after data alignment. Therefore, for each client, the intermediate output, $\textbf{H}^{(i)}_b $ with $b \in [1, M/B]$, is a \emph{fixed mapping} to the local dataset. Hence it suffices for the clients to compute the forward pass on all the batches \emph{once} during the \emph{first} epoch and secret-share the intermediate output (and the labels) to the servers. No further \emph{upstream communications} are needed afterwards from the clients to the servers. The servers can store the intermediate outputs and reuse them for future training epochs, thereby reducing the communication overhead.

Thereafter, each client $E_i$ sends the secret-shares of {its} intermediate output $\textbf{H}^{(i)}_b$, which is denoted by $\ash{\textbf{H}^{(i)}_b}$ for $b \in [1, M/B]$, to the servers. The label client also secret-shares the labels $\textbf{y}_b$ as $\ash{\textbf{y}_b}$. {Upon receiving} all batches of secret-shares from each client, the servers \emph{vertically} concatenate the shares of the intermediate outputs from all clients for each batch respectively:
$\ash{{\textbf{H}}_b}=\{\ash{\textbf{H}^{(1)}_b},\dots, \ash{\textbf{H}^{(N)}_b}\}$ for $b \in [1, M/B]$ and then \emph{horizontally} concatenate all $\ash{{\textbf{H}}_b}$ into $\ash{\overline{\textbf{H}}} = \{\ash{{\textbf{H}}_1};\dots;\ash{{\textbf{H}}_{M/B}} \}$. 
All $\ash{\textbf{y}_b}$ are also \emph{horizontally} concatenated into $\ash{\overline{\textbf{y}}} = \{\ash{{\textbf{y}}_1};\dots;\ash{{\textbf{y}}_{M/B}} \}$.

\paragraph{Secure Gradients Computation for Global Model}~During global model training, at the beginning of \emph{every} epoch, the servers make a blackbox call to a \emph{secure shuffling} protocol, \SHU, on $\ash{\overline{\textbf{H}}}$ and $\ash{\overline{\textbf{y}}} $, to obliviously shuffle the rows and obtain $\ash{\overline{\textbf{H}}^*}$ and $\ash{\overline{\textbf{y}}^*}$. Both $\ash{\overline{\textbf{H}}^*}$ and $\ash{\overline{\textbf{y}}^*}$ are further divided into $M/B$ batches again, as $\ash{\overline{\textbf{H}}^*_b}$ and $\ash{\overline{\textbf{y}}^*_b}$ with $b \in [1, M/B]$ which are then \emph{sequentially} fed into the global model, one batch at a time. This shuffling procedure enables subsampling without replacement.

As a result, for \emph{each} {step} $t$~(note the corresponding batch index in the current epoch is given by $b:= t\mod{B_\emph{num}} $), the servers compute the logits $\ash{\boldsymbol{Z}}$ and hence the per-sample loss $\ash{ \mathcal{L}_j}$ with $j \in [1, B]$, by calling the forward-passing module \FWD. Thereafter, the per-sample gradients $\ash{\boldsymbol{g}_{\boldsymbol{\theta}, j}} = \ash{{\partial \mathcal{L}_j}/{\partial \boldsymbol{\theta} }}$ for the global model are computed through the back-propagation module \BWD. 

\paragraph{Noise Addition for Securely Computed Information} Once the per-sample gradients are computed, {they} are then clipped into 
\begin{align}
\ash{\overline{\boldsymbol{g}}_{\boldsymbol{\theta}, j}} = \left[\!\left[{\boldsymbol{g}_{\boldsymbol{\theta}, j}}/{\max(1, {\| \boldsymbol{g}_{\boldsymbol{\theta}, j}\|_2}/{\gamma} ) } \right]\!\right]~~, \label{clipping}
\end{align}
where the clipping threshold $\gamma$ is set to the maximum $\ell_2$ norm bound on each gradient. We use $\PNM, \PDiv$ to compute the norm and the clipping factors.
The clipped gradients in a mini-batch are aggregated locally by the servers into the aggregated gradient $\ash{\overline{\boldsymbol{g}}_{\boldsymbol{\theta}}} = \sum^B_{j = 1} \ash{\overline{\boldsymbol{g}}_{\boldsymbol{\theta}, j}}$.  
{Note in the} pre-processing phase, the servers run \PGS~with the input of $(T_\emph{num}, n_{\boldsymbol{\theta}}, \sigma \cdot\gamma)$, where $n_{\boldsymbol{\theta}}$ is the number of parameters in $\boldsymbol{\theta}$, to generate a secret-shared noise matrix $\ash{\boldsymbol{\mathcal{N}}_{\text{Tab}} }$. Then, at the step $t$, the noise vector $\ash{\boldsymbol{\mathcal{N}}_{\boldsymbol{\theta}}} = \ash{\boldsymbol{\mathcal{N}}_{\text{Tab}}}_{[t,:]}$ is added {to} the aggregated gradient, which results in the privatized gradient  
\[
\ash{\tilde{\boldsymbol{g}}_{\boldsymbol{\theta}}} = \frac{1}{B} (\ash{\overline{\boldsymbol{g}}_{\boldsymbol{\theta}}} + \ash{\boldsymbol{\mathcal{N}_{\boldsymbol{\theta}}}})~~.
\]
The servers update $\ash{\globalm}$ through gradient descent with gradient $ \ash{\tilde{\boldsymbol{g}}_{\theta}}$. The formal description of the protocol is shown in~\figref{pfirst1}.


\begin{protofig}{Protocol $\PFirstA$}{Training a Global model with Frozen Local models under Gaussian Mechanism with Shuffling}{pfirst1}

\textbf{Parameters:} Clients $\pclients = E_1, \ldots, E_N$, aggregation servers $\pservers = \{P_k\}^K_{k = 1}$. 

\textbf{Preprocessing:} The servers run $\PGS$ with input $(T_\emph{num}, n_{\boldsymbol{\theta}}, \sigma \gamma)$ to receive $\ash{\boldsymbol{\mathcal{N}}_{\text{Tab}} }$. \\ 

\textbf{Client-Server Phase:} Each client $E_i$ does the following:~Before the global model training starts, 
\begin{enumerate}
    \item Locally computes $\textbf{H}^{(i)}_b = \{ \interout^{(i)}_{b, 1};\dots;\interout^{(i)}_{b, B}\}$ where $\interout^{(i)}_{b, j} = \boldsymbol{\mathcal{F}}_i(\boldsymbol{\phi}_i; \boldsymbol{x}^{(i)}_{b, j})$, and secret-shares $\ash{\textbf{H}^{(i)}_b}$, for $b \in [1, M/B ]$, then sends all the batches to $\pservers$.
    \item The label client $E_N$ also secret-shares $\ash{\textbf{y}_b}$ for $b \in [1, M/B ]$ to $\pservers$. \\
\end{enumerate}

\textbf{Server MPC: }~At the beginning of each epoch, $\pservers$ do the following:
\begin{enumerate}
    \item Only if at the beginning of the \emph{first} epoch, vertically concatenate $\ash{\textbf{H}^{(i)}_b}$ received from $\pclients$ for $i \in [1, N]$, which results in $\ash{\textbf{H}_b}$ with $b \in [1, M/B ]$, and then horizontally concatenate $\ash{\textbf{H}_b}$ and $\ash{\textbf{y}_b}$, for $b \in [1, M/B ]$, to result in$\ash{\overline{\textbf{H}}}$ and $\ash{\overline{\textbf{y}}}$, respectively. 
    \item Run $\SHU$ on $\ash{\overline{\textbf{H}}}$ and $\ash{\overline{\textbf{y}}}$, to receive the sorted values $\ash{\overline{\textbf{H}}^{*}}$ and $\ash{\overline{\textbf{y}}^{*}}$ respectively. 
    \item Split $\ash{\overline{\textbf{H}}^{*}}$ and $\ash{\overline{\textbf{y}}^{*}}$ into $M/B$ mini-batches as $\ash{\overline{\textbf{H}}^{*}_b}$ and $\ash{\overline{\textbf{y}}^{*}_b}$ with $b \in [1, M/B]$. 
\end{enumerate}

For each step in the current epoch, $\pservers$ do the following:
\begin{enumerate}
    \item Compute per-sample loss $\ash{\mathcal{L}_j}$ for $j \in [1, B]$, by inputting $\ash{\overline{\textbf{H}}^{*}_b}$, $\ash{\globalm}$ and $\ash{\overline{\textbf{y}}^{*}_b}$ into $\FWD$.
    \item Compute $\ash{\boldsymbol{g}_{\boldsymbol{\theta}, j}} = \ash{\frac{\partial \mathcal{L}_j}{\partial \boldsymbol{\theta}}}$, for $j \in [1, B]$, by running $\BWD$.
    \item Compute $\ash{\gamma_j} = \msf{max}(1, \text{n}/\gamma)$, where $\ash{n} = \PNM(\ash{\boldsymbol{g}_{\boldsymbol{\theta}, j}})$, for $j \in [1, B]$.
    \item Finally, obtain $\ash{\overline{\boldsymbol{g}}_{\boldsymbol{\theta}, j}} = \ash{\frac{\boldsymbol{g}_{\boldsymbol{\theta}, j}}{\gamma_j}}$, by running $\PDiv$.
    \item Locally aggregate the gradients by $\ash{\overline{\boldsymbol{g}}_{\boldsymbol{\theta}}} = \sum^B_{j = 1} \ash{\overline{\boldsymbol{g}}_{\boldsymbol{\theta}, j}}$.
    \item Locally add noise by $\ash{ \tilde{\boldsymbol{g}}_{\theta}} = \frac{1}{B}(\ash{\overline{\boldsymbol{g}}_{\bs{\theta}}} + \ash{\boldsymbol{\mathcal{N}}_{\boldsymbol{\theta}}})$, where $\ash{\boldsymbol{\mathcal{N}}_{\boldsymbol{\theta}}} = \ash{\boldsymbol{\mathcal{N}}_{\text{Tab}}}_{[t,:]}$.
    \item Update $\ash{\globalm}$ via gradient descent on $\ash{ \tilde{\bs{g}}_{\theta}}$.
\end{enumerate}
    
\end{protofig}

\subsubsection{Protocol With Banded Matrix Factorization}
\label{sec:mpc_dp_localfreez_bmp}
In~\ref{sec:mpc_dp_localfreez_shuffle}, we design a protocol that achieves privacy amplification by subsampling via secure shuffling over the data secret-shared from the clients. However, for {large-scale} datasets, performing a secure shuffle for each epoch could take up a considerable amount of time. Thus, we propose an alternative protocol that does not rely on the subsampling but still achieves {a} competitive privacy-utility trade-off. This protocol reduces the communication and computation in the online phase by using a public correlation matrix. 
In particular, our protocol utilizes the Banded Matrix Factorization~(BandMF) mechanism~\cite{Kalinin2024banded}, which leverages \emph{banded correlation matrices} to generate \emph{correlated noise} to be applied during training.

We first follow the same concept of~$\text{b}$-\emph{min-separated participation} schema {proposed in}~\cite{Choquette-Choo2023amplified}. Under this schema, if a single item (such as a sample) contributed to a model gradient vector at step $t$, then the earliest it can contribute again is at step $t+\text{b}$. 
Our VFL training protocol naturally satisfies the schema:~The order of the samples at each client is fixed, and therefore the interval between two adjacent {participations} of any sample point is constant. 
In fact, our training protocol also satisfies~\emph{fixed-epoch-order participation} schema~(the $(\kappa, \text{b})$-\emph{participation})~\cite{ChoquetteChoo2023multi},  a stricter participation schema where each sample participates $\kappa$ times with the constraint that any two adjacent {participations} are exactly \text{b} steps apart.
{ Note we fix $\kappa = E_\emph{num} $ and $\text{b} = B_\emph{num}$ throughout this paper}.

In \emph{Matrix Factorization} mechanisms~\cite{denisov2022improved,ChoquetteChoo2023multi}, let $\boldsymbol{A} \in \mathbb{R}^{T_\emph{num} \times T_\emph{num}}$ be an appropriate \emph{workload matrix}, and $\boldsymbol{G}=\{\boldsymbol{g}^1_{{\boldsymbol{\theta}}};\dots;\boldsymbol{g}^{T_\emph{num}}_{\boldsymbol{\theta}} \} $ is a stream of model batch gradients. Here each batch gradient $\boldsymbol{g}^t_{\boldsymbol{\theta}}$ has a bounded $\ell_2$ norm, as $\| \boldsymbol{g}^t_{\boldsymbol{\theta}} \| \leq \gamma$. The matrix factorization on $\boldsymbol{A}$ is represented by $\boldsymbol{A}=\boldsymbol{B}_{\text{MF}}\boldsymbol{C}_{\text{MF}}$, which is used to privately estimate the quantity of $\boldsymbol{A}\boldsymbol{G}$, as   
\begin{align}
\widehat{\boldsymbol{A}\boldsymbol{G}}=\boldsymbol{B}_{\text{MF}}(\boldsymbol{C}_{\text{MF}}\boldsymbol{G} + \boldsymbol{Z}) = \boldsymbol{A}(\boldsymbol{G} + \boldsymbol{C}^{-1}_{\text{MF}}\boldsymbol{Z})~~, \label{eq:mf_eq}
\end{align}
where $\boldsymbol{Z}$ is an appropriately scaled Gaussian noise vector whose scale is determined by $\text{sens}(\boldsymbol{C}_{\text{MF}})$, which is the sensitivity of $\boldsymbol{C}_{\text{MF}}$ defined in Definition~1 in~\cite{ChoquetteChoo2023multi} and also summarized in \ref{def:sens_MF}. Typically, $\boldsymbol{C}_{\text{MF}}$ is also called a \emph{query matrix} or \emph{encoder} since it encodes $\boldsymbol{G} $ as $\boldsymbol{C}_{\text{MF}}\boldsymbol{G} $, which is expected to be made private. 

Recently, by utilizing the banded root square factorization~\cite{Kalinin2024banded}, particularly for {the} SGD algorithm, the SGD {workload matrix} $\boldsymbol{A}$ can be further constructed through $\boldsymbol{A} = \eta \boldsymbol{A}_{\alpha, \beta}$
where $\eta$ is the learning rate, and $\boldsymbol{A}_{\alpha, \beta}$ is a lower triangular Toeplitz-matrix. 
Based on Theorem~1 in~\cite{Kalinin2024banded}, given any value $p \in [1, T_\emph{num}]$,  the $p$-\emph{banded square-root}~(BSR) of $\boldsymbol{A}_{\alpha, \beta}$:~$C^{|p|}_{\alpha, \beta}$, is computed through
\begin{align}
C^{|p|}_{\alpha, \beta} = \text{LDToep}(c_0,\dots, c_{{T_\emph{num}}-1})~~, \nonumber 
\end{align}
where $c_j =  \sum^{j}_{i=0}\alpha^{j-i}r_{j-i}r_i\beta^{i}$ with $r_i = |\binom{-1/2}{i} | $ for $0<j<p$, 
$\alpha$ is the weight decay parameter and $\beta$ is the momentum strength. Note here $c_0 = 1$ and $c_j = 0$ when $j \geq p$. This $p$-BSR matrix serves as the {query matrix} for computing the correlated noise
and is treated as a public parameter, computed {during} the pre-processing phase and sent in plaintext to the servers.
We use $\boldsymbol{\Omega}$ to denote the $p$-BSR matrix for simplicity. Moreover, we set $p = \text{b}$ for repeated data participation. Under \text{b}-min-separated-participation, the necessary amount of noise depends on the \emph{sensitivity} of $\boldsymbol{\Omega}$, $\text{sens}_{\kappa,\text{b}}(\boldsymbol{\Omega})$, which is defined in Theorem~2 in~\cite{Kalinin2024banded} and also summarized in~\ref{def:sens}. 

After pre-processing, our protocol works {as follows}:~{Similar to} the protocol in~Section~\ref{sec:mpc_dp_localfreez_shuffle}, we still {describe the protocol using the same three main components}.

\paragraph{Forward-passing and Secret-sharing} Each client still passes all batches of local data to obtain the intermediate outputs $\textbf{H}^{(i)}_b$ for $b \in [1, M/B]$ at the \emph{beginning} and then secret-shares the intermediate outputs to the servers. Similarly, $\textbf{y}_b$ with $b \in [1, M/B]$ are also secret-shared to the servers. Since all local {models} are frozen, this upstream communications are still only done once before the \emph{first} epoch. {After receiving} all uploads, the servers vertically concatenate the shares of the intermediate outputs from all clients for every batch $b$ to obtain $\ash{{\textbf{H}}_b}$ with $b \in [1, M/B]$. Unlike~Section~\ref{sec:mpc_dp_localfreez_shuffle}, note here the servers \emph{will not} shuffle the horizontally concatenated version of $\ash{{\textbf{H}}_b}$ and $\ash{\textbf{y}_b}$ by $b \in [1, M/B]$, which are $\ash{\overline{\textbf{H}}}$ and $\ash{\overline{\textbf{y}}}$, but store $\ash{{\textbf{H}}_b}$ and $\ash{\textbf{y}_b}$ {in} their original order from $b =1~\text{to}~M/B$ to align with the participation pattern defined before. 

\paragraph{Secure Gradients Computation for Global Model} In the training phase, for \emph{each} {step} $t$, the servers feed $\ash{{\textbf{H}}_b}$ into the global model for conducting forward-passing through \FWD. Once the losses for all samples are computed, the back-propagation is conducted by \BWD, and the per-sample gradients for the global model $\ash{\boldsymbol{g}_{\boldsymbol{\theta}, j}} $ with $j \in [1, B]$ are computed. 

\paragraph{Noise Addition for Securely Computed Information} After clipping the per-sample gradients, which is conducted through Eq.~\eqref{clipping}, the aggregated gradient $\ash{\overline{\boldsymbol{g}}_{\boldsymbol{\theta}}} = \sum^B_{j = 1} \ash{\overline{\boldsymbol{g}}_{\boldsymbol{\theta}, j}}$ is obtained.  Unlike~Section~\ref{sec:mpc_dp_localfreez_shuffle}, the servers now add correlated {noise} to the aggregated gradient under {the} BandMF mechanism. During the pre-processing phase, the servers run the noise sampling protocol \PDP~with input $(T_\emph{num}, n_{\boldsymbol{\theta}}, \sigma \cdot \gamma \cdot \text{sens}_{\kappa,\text{b}}(\boldsymbol{\Omega}))$ to generate $\ash{\boldsymbol{\mathcal{N}}_{\text{Tab}}}$. Then, the servers compute $\ash{\boldsymbol{\mathcal{N}}_{\text{Corr}}}= \boldsymbol{\Omega}^{-1}\ash{\boldsymbol{\mathcal{N}}_{\text{Tab}}}$. 
Hence at the step $t$, the $t$-th row of $\ash{\boldsymbol{\mathcal{N}}_{\text{Corr}}}$ is added to the aggregated gradient, i.e., $\ash{\boldsymbol{\mathcal{N}}_{\boldsymbol{\theta}}} = \ash{\boldsymbol{\mathcal{N}}_{\text{Corr}}}_{[t,:]}$. Eventually, the servers compute the privatized gradient through $\ash{\tilde{\boldsymbol{g}}_{\boldsymbol{\theta}}} = \frac{1}{B} (\ash{\overline{\boldsymbol{g}}_{\boldsymbol{\theta}}} + \ash{\boldsymbol{\mathcal{N}_{\boldsymbol{\theta}}}})
$, which is then used to update the global model $\ash{\globalm}$. The formal description of this BandMF mechanism based protocol is shown in \figref{pfirst2}.

\begin{protofig}{Protocol $\PFirstB$}{Training a Global model with Frozen Local models under Banded Matrix Factorization Mechanism}{pfirst2}

\textbf{Parameters:} Clients $\pclients = E_1, \ldots, E_N$, aggregation servers $\pservers = \{P_k\}^K_{k = 1}$. \\
\textbf{Preprocessing:} The servers run $\PGS$ with $(T_\emph{num}, n_{\boldsymbol{\theta}}, \sigma \gamma \text{sens}_{\kappa,\text{b}}(\boldsymbol{\Omega}))$ to receive $\ash{\boldsymbol{\mathcal{N}}_{\text{Tab}} }$. The servers compute $\ash{\boldsymbol{\mathcal{N}}_{\text{Corr}}}=\boldsymbol{\Omega}^{-1}\ash{\boldsymbol{\mathcal{N}}_{\text{Tab}}}$. \\

\textbf{Client-Server Phase:} Proceeds the same steps as in \figref{pfirst1}, with the clients sending $\ash{\textbf{H}^{(i)}_b}$ and $\ash{\textbf{y}_b}$ for $b \in [1, M/B ]$ to $\pservers$. \\


\textbf{Server MPC:}~At the beginning of each epoch, $\pservers$ do the following:
\begin{enumerate}
    \item Only if at the beginning of the \emph{first} epoch, vertically concatenate $\ash{\textbf{H}^{(i)}_b}$ received from $\pclients$ for $i \in [1, N]$, denoted by $\ash{\textbf{H}_b}$ with $b \in [1, M/B ]$.
    \end{enumerate}
    
For each step in the current epoch, $\pservers$ do the following:
    \begin{enumerate}
    \item Given $\ash{{\textbf{H}}_b}$ and $\ash{\textbf{y}_b}$, follow the same steps from 1 to 5, as described in \figref{pfirst1}, to obtain $\ash{\overline{\boldsymbol{g}}_{\bs{\theta}}}$.
    \item  Locally add noise to the aggregated gradients as $\ash{ \tilde{\boldsymbol{g}}_{\theta}} = \frac{1}{B}(\overline{\boldsymbol{g}}_{\bs{\theta}} + \ash{\boldsymbol{\mathcal{N}}_{\boldsymbol{\theta}}})$, where $\ash{\boldsymbol{\mathcal{N}}_{\boldsymbol{\theta}}} = \ash{\boldsymbol{\mathcal{N}}_{\text{Corr}}}_{[t,:]}$.
    \item Update $\ash{\globalm}$ via gradient descent on $\ash{ \tilde{\bs{g}}_{\theta}}$.
\end{enumerate}
    
\end{protofig}

\subsection{Training a Global and the Local models with DP}
\label{sec:mpc_dp_gl_update}

The two protocols introduced in Section~\ref{sec:mpc_dp_localfreez} enable privacy-preserving training of the global model while freezing the local models. In practice, however, clients may also have the incentive to update their local models:~A client may not always have a good pre-trained model for all tasks to start with; meanwhile, clients may also want to gain some insights from participating {in} the VFL for future tasks. Also, as mentioned earlier, directly revealing per-sample gradients from servers would leak too much private information to the clients~\cite{zhu2019leakage}. These practical needs motivate us to design an alternative protocol that supports the update of local models while preserving privacy. 

However, na\"ively applying privacy-enhancing techniques to compute the chain-rule for local model updates in Section~\ref{sec:plian} may be infeasible. Notice that a client needs per-sample gradients of the loss w.r.t. its local model's output from the servers. If the servers publish the per-sample gradients and add DP-preserving noise to each gradient, the amount of noise added will be {so tremendously} large that the utility is diminished. On the other hand, if the client delegates the local per-sample gradient computation to the servers via MPC, the computation and communication overhead will scale with the size of the local model, which is impractical. 

In this section, we introduce a novel training protocol that bypasses the limitation of directly releasing {per-sample gradients}. The servers release privacy-budget friendly information and the clients can approximate the per-sample gradients from the released information. Note {that} we will still employ the BandMF mechanism~\cite{Kalinin2024banded}, as both upstream and downstream communications between the servers and clients are required at each step in this protocol. This enforces consistency between the information uploaded to the servers and the information received by the clients, ensuring that they correspond to the same mini-batch, which ultimately \emph{precludes} the anonymity of sample participation. We follow the same pre-processing procedure described in Section~\ref{sec:mpc_dp_localfreez_bmp} to predefine the participation pattern and compute the $p$-BSR matrix $\boldsymbol{\Omega}$.

{Unlike} the previous protocols, our full protocol consists of four main components: 1) forward-passing and secret-sharing for intermediate outputs and auxiliary information, 2) secure computation of gradients for the global model and selective layer(s) of local models, 3) privacy-preserving noise addition for the securely computed information including the global model update and auxiliary information for local gradients reconstruction, and 4) reconstruction of per-sample gradient and update of local models. 

\paragraph{Forward-passing and Secret-sharing} At each training step $t$, each client $E_i$, in addition to secret-sharing $\textbf{H}^{(i)}_b$ to the servers, also \emph{locally} computes per-sample gradients, ${\partial \interout^{(i)}_{b, j}}/{{\partial \boldsymbol{\phi}}^{L}_i}$, for $j \in [1, B]$, with respect to the model parameters ${\boldsymbol{\phi}}^{L}_i$ of a client chosen layer(s) $L$ of the local model, and secret-shares those gradients to the servers. These gradients, which {are} viewed as auxiliary information, will be further processed to construct an equation system for local model gradients estimation. 
Note {that} ${\boldsymbol{\phi}}^{L}_i$ is allowed to be different at each client. For example, one client could choose the last fully connected~(fc) layer, while another client could choose the last LoRA layer~\cite{hu2022lora}.
The servers then compute the loss $\ash{ \mathcal{L}_j}$ for $j \in [1, B]$, through the same steps in previous protocols. 

\paragraph{Secure Gradients Computation for Global/Local Models} First, the servers still securely compute the per-sample gradient of the global model $\ash{\boldsymbol{g}_{\boldsymbol{\theta}, j}}$ for $ j \in [1,B]$ via~\BWD~as in previous protocols. In addition, the servers securely compute the per-sample gradients for each client's local layer(s) $\ash{\boldsymbol{g}_{\boldsymbol{\phi}^{L}_{i}, j}}$ for $ j \in [1,B]$ using \PMult:~We have shown in Section~\ref{sec:mpc_dp_localfreez} that the servers and the clients can collaborate to securely compute the gradient w.r.t. the intermediate output of the clients. A similar procedure can be used to compute the gradient of loss w.r.t. model parameters in a local model. 
Note as the client has locally computed ${\partial \interout^{(i)}_{b, j}}/{{\partial \boldsymbol{\phi}}^{L}_i}$ and secret-shares the results to the servers during forward-passing, by chain rule, the per-sample gradient w.r.t. $\boldsymbol{\phi}^{L}_i$ can be securely computed through
    \begin{align}
    \ash{\boldsymbol{g}_{\boldsymbol{\phi}^{L}_{i}, j}} = \ash{\frac{\partial  \mathcal{L}_j}{\partial \boldsymbol{\phi}^{L}_i}} = \ash{ \frac{\partial\mathcal{L}_j}{\partial  \interout^{(i)}_{b, j}}} \cdot \ash{ \frac{\partial \interout^{(i)}_{b, j}}{{\partial \boldsymbol{\phi}}^{L}_i}}~~. \label{eq:multip}
    \end{align}

\paragraph{Noise Addition for Securely Computed Information} After gradient computation, our protocol adds privacy-preserving noise to the gradients before releasing them. For per-sample gradient clipping, the servers \emph{concatenate} all the gradients under the $j$-th sample~(after flattening), which results in:
\begin{align}
\ash{\boldsymbol{g}_{j}} = \ash{\{\boldsymbol{g}_{\boldsymbol{\theta}, j}, \boldsymbol{g}_{\boldsymbol{\phi}^{L}_{1}, j}, \dots, \boldsymbol{g}_{\boldsymbol{\phi}^{L}_{N}, j} \} }~~, \nonumber 
\end{align}
where $\boldsymbol{g}_{j} \in \mathbb{R}^{n_{\text{conca}}}$, with $n_{\text{con}} = (n_{\boldsymbol{\theta}} + \sum^N_{i=1}n^{L}_{i})$, where $n^{L}_{i}$ is the number of parameters in ${\boldsymbol{\phi}}^{L}_i$. Then, servers apply clipping on $\boldsymbol{g}_{j}$: 
\begin{align}
\ash{\overline{\boldsymbol{g}}_{j}} = \ash{{\boldsymbol{g}_{j}}/{\gamma_j}}~~,
\end{align}
where $\gamma_j = \max\left(1, \frac{\| \boldsymbol{g}_{j}\|_2}{\gamma}\right)$ for $j \in [1, B]$.
After clipping, the servers aggregate all clipped gradients to obtain $\ash{\overline{\boldsymbol{g}}} = \sum^B_{j = 1} \ash{\overline{\boldsymbol{g}}_{j}}$ and add the correlated noise under the BMF mechanism to the aggregated gradient by $\ash{\Tilde{\boldsymbol{g}}} = \frac{1}{B}(\ash{\overline{\boldsymbol{g}}} + \ash{\boldsymbol{\mathcal{N}}_{\Tilde{\boldsymbol{g}}} })$
where $\boldsymbol{\mathcal{N}}_{\Tilde{\boldsymbol{g}} }$ is obtained through $\ash{\boldsymbol{\mathcal{N}}_{\Tilde{\boldsymbol{g}}}} = \left[\boldsymbol{\Omega}^{-1}\ash{\boldsymbol{\mathcal{N}}_{\text{Tab}}}\right]_{[t,:]} $ and that $\ash{\boldsymbol{\mathcal{N}}_{\text{Tab}}}$ is computed by running \PGS~with the input of $(T_{\emph{num}}, n_{\text{conca}}, \sigma \cdot \gamma \cdot \text{sens}_{\kappa,\text{b}}(\boldsymbol{\Omega}))$. After obtaining $ \ash{\Tilde{\boldsymbol{g}}} $, servers use $\ash{\Tilde{\boldsymbol{g}}_{\boldsymbol{\theta}}}$, which is the gradient of global model part, to update the global model $\ash{\globalm}$. Meanwhile, the servers take the gradient of the chosen local layer(s) for each client  $\Tilde{\boldsymbol{g}}_{\boldsymbol{\boldsymbol{\phi}^{L}_{i}}}$ out from $\ash{\Tilde{\boldsymbol{g}}}$, and distribute the gradient back to $E_i$, for $i \in [1, N]$.

\paragraph{Local Per-sample Gradient Reconstruction and Model Updates} On the client's end, since $\Tilde{\boldsymbol{g}}_{\boldsymbol{\phi}^{L}_{i}}$ is the privatized gradient~(aggregated at batch-level) with noise added, it cannot be directly used to conduct backpropagation on {the} local model: 
Recall that in order to recover the back-propagation flow in the plain-text model in Section~\ref{sec:plian}, the client needs $\{\partial\mathcal{L}_j/{\partial  \interout^{(i)}_{b,j}}\}^{B}_{j=1}$, i.e., the per-sample gradient w.r.t. the intermediate output from the client. Although recovering the exact per-sample gradient is infeasible as the received $\Tilde{\boldsymbol{g}}_{\boldsymbol{\phi}^{L}_{i}}$ is privacy-preserved, we show that a client can estimate the clipped per-sample gradient w.r.t. its intermediate output. In particular, 
We can establish the following linear equation system, based on the chain-rule applied on  $\Tilde{\boldsymbol{g}}_{\boldsymbol{\phi}^{L}_{i}}$:  

\begin{align}
& \Tilde{\boldsymbol{g}}_{\boldsymbol{\phi}^{L}_{i}}=\left[ \begin{matrix}
\Tilde{\boldsymbol{g}}_{\boldsymbol{\phi}^{L}_{i}}[1] \\
                             \vdots \\
                              \Tilde{\boldsymbol{g}}_{\boldsymbol{\phi}^{L}_{i}}[n^{L}_{i}] 
\end{matrix}
  \right] =
\frac{1}{B}\left[ \begin{matrix}
\frac{\partial \interout^{(i)}_{b, 1}}{{\partial \boldsymbol{\phi}}^{L}_i[1]}  &
\frac{\partial \interout^{(i)}_{b, 2}}{{\partial \boldsymbol{\phi}}^{L}_i[1]}  & 
\dots & \frac{\partial \interout^{(i)}_{b, B}}{{\partial \boldsymbol{\phi}}^{L}_i[1]}  \\
    \vdots & \vdots & \vdots & \vdots \\
   \frac{\partial \interout^{(i)}_{b, 1}}{{\partial \boldsymbol{\phi}}^{L}_i[n^{L}_{i}]} &  
   \frac{\partial \interout^{(i)}_{b, 2}}{{\partial \boldsymbol{\phi}}^{L}_i[n^{L}_{i}]} & 
   \dots &   \frac{\partial \interout^{(i)}_{b, B}}{{\partial \boldsymbol{\phi}}^{L}_i[n^{L}_{i}]}
  \end{matrix}
  \right] \cdot \nonumber \\ 
  & \qquad \qquad \qquad  \qquad \ \ \ \   \overline{\boldsymbol{\Gamma}} \cdot \left[ \begin{matrix} 
  \frac{\partial\mathcal{L}_1}{\partial  \interout^{(i)}_{b, 1}} \\
  \vdots \\
  \frac{\partial\mathcal{L}_B}{\partial  \interout^{(i)}_{b, B}}
\end{matrix}
  \right]+ \frac{1}{B}\left[ \begin{matrix}
                            \mathcal{N}_i[1] \\
                             \vdots \\
                             \mathcal{N}_i[n^{L}_{i}] 
\end{matrix}
  \right] \nonumber \\
& \qquad \qquad \qquad  \quad \ \ = \frac{1}{B}({\interout}_{{\boldsymbol{\phi}}^{L}_i} \cdot \overline{\boldsymbol{\Gamma}} \cdot \boldsymbol{g}_{\textbf{H}^{(i)}_b} + \boldsymbol{\mathcal{N}}_i) \nonumber\\
& \qquad \qquad \qquad  \quad \ \ = \frac{1}{B}({\interout}_{{\boldsymbol{\phi}}^{L}_i} \cdot  \widehat{\boldsymbol{g}}_{\textbf{H}^{(i)}_b} + \boldsymbol{\mathcal{N}}_i)~~, \label{lineq}
\end{align}
where matrix ${\interout}_{{\boldsymbol{\phi}}^{L}_i} = \left\{{\partial \interout^{(i)}_{b, j}}/{{\partial \boldsymbol{\phi}}^{L}_i[l]}\right\}^B_{j=1} \in \mathbb{R}^{n^{L}_{i} \times (d_{\interout^{(i)}} B)}$ with $l \in [1, n^L_i]$, and vector $\boldsymbol{g}_{\textbf{H}^{(i)}_b} = \left\{{\partial\mathcal{L}_j}/{\partial  \interout^{(i)}_{b, j}}\right\}^B_{j=1} \in \mathbb{R}^{d_{\interout^{(i)}} B}$. More specifically,
\begin{align}
{\partial \interout^{(i)}_{b, j}}/{{\partial \boldsymbol{\phi}}^{L}_i[l]} = \left[{\partial \interout^{(i)}_{b, j}[1]}/{{\partial \boldsymbol{\phi}}^{L}_i[l]},\dots, {\partial \interout^{(i)}_{b, j}[d_{\interout^{(i)}}]}/{{\partial \boldsymbol{\phi}}^{L}_i[l]}\right]~~, \nonumber
\end{align}
and 
\begin{align}
{\partial\mathcal{L}_j}/{\partial  \interout^{(i)}_{b, j}} = \left[ {\partial\mathcal{L}_j}/{\partial  \interout^{(i)}_{b, j}[1]},\dots,{\partial\mathcal{L}_j}/{\partial  \interout^{(i)}_{b, j}[d_{\interout^{(i)}}]} \right]^{T}~~. \nonumber 
\end{align}
The term $\overline{\boldsymbol{\Gamma}} \in \mathbb{R}^{(d_{\interout^{(i)}}B) \times (d_{\interout^{(i)}} B)}$ is a diagonal matrix: $\overline{\boldsymbol{\Gamma}}=\text{diag}(\boldsymbol{\Gamma}_j )$
where each diagonal entry, ${\boldsymbol{\Gamma}}_j \in \mathbb{R}^{d_{\interout^{(i)}} \times d_{\interout^{(i)}}}$, is also a diagonal matrix whose diagonal entries are $\gamma^{-1}_j$, i.e., $\boldsymbol{\Gamma}_j = \text{diag}(\gamma^{-1}_j)$. The term $\boldsymbol{\mathcal{N}}_i \in \mathbb{R}^{n^{L}_{i} }$ is a sub-vector in $\boldsymbol{\mathcal{N}}_{\Tilde{\boldsymbol{g}} }$ which is the noise that was added to privatize $\Tilde{\boldsymbol{g}}_{\boldsymbol{\phi}^{L}_{i}}$. 

The key intuition is that the equation in Eq.~\eqref{lineq} is a linear system with noisy measurement:~We can estimate $\widehat{\boldsymbol{g}}_{\textbf{H}^{(i)}_b}$, which is
\begin{align}
\widehat{\boldsymbol{g}}_{\textbf{H}^{(i)}_b} \triangleq \overline{\boldsymbol{\Gamma}} \cdot \boldsymbol{g}_{\textbf{H}^{(i)}_b} = \{\gamma_j \cdot {\partial\mathcal{L}_j}/{\partial  \interout^{(i)}_{b,j}}\}^{B}_{j=1}~~,
\end{align}
from the locally computed ${\interout}_{{\boldsymbol{\phi}}^{L}_i}$ and the observed $\Tilde{\boldsymbol{g}}_{\boldsymbol{\phi}^{L}_{i}}$. Since the gradients are estimated from privacy-preserving outputs from the servers, the client can apply standard solvers such as LS~(least squares), MMSE~(minimum mean square error), RR (ridge regression) or any other linear estimation techniques locally to solve  Eq.~\eqref{lineq} under {the} \emph{plaintext model}. {Note our protocol is estimator-agnostic. We show how estimation techniques can extract per-sample gradients from the aggregated private gradients to enable backpropagation, without relying on any specific estimator. Exploring alternative estimators and their potential performance gains is left for future work.}

Once $\widehat{\boldsymbol{g}}_{\textbf{H}^{(i)}_b}$ is estimated, the $i$-th client can update its local layers by back-propagation. The complete protocol is presented in \figref{psecond}.

\subsubsection{Estimation Performance Analysis}~Typically, to obtain a stable and reliable estimate in a linear system with noise, it is desirable to have at least as many observations as unknowns. Under this condition, the system becomes either exactly determined or overdetermined. In our case, this \emph{translates} to $n^{L}_{i} \geq d_{\interout^{(i)}} \cdot B$, which ensures that the estimation problem remains well-posed and can be solved robustly. In general, aside from the ratio between $n^{L}_{i}$ and $d_{\interout^{(i)}} \cdot B$, the estimation error is dependent on the intrinsic properties of the matrix ${\interout}_{{\boldsymbol{\phi}}^{L}_i}$, the characteristics of the added noise, and the choice of estimator as well as any assumptions on $\widehat{\boldsymbol{g}}_{\textbf{H}^{(i)}_b}$ if needed. Although various estimators could be used in the linear system, we take RR estimator, which is defined by
\begin{align}
\widehat{\boldsymbol{g}}^{\text{ridge}}_{\textbf{H}^{(i)}_b} = ({\interout}^T_{{\boldsymbol{\phi}}^{L}_i}{\interout}_{{\boldsymbol{\phi}}^{L}_i} + \lambda I )^{-1}{\interout}^T_{{\boldsymbol{\phi}}^{L}_i}\Tilde{\boldsymbol{g}}_{\boldsymbol{\phi}^{L}_{i}} ~~, \label{ridge_estimation}
\end{align}
where $\lambda$ is regularization parameter, as an illustrative example, as RR is more stable than LS in ill-conditioned systems and also unlike MMSE, does not require a prior on $\widehat{\boldsymbol{g}}_{\textbf{H}^{(i)}_b}$'s distribution, while still allowing bounded-error analysis. As a common heuristic in RR, we set $\lambda$ {to} the noise {variance}. Moreover, as the $\ell_2$ norm of $\widehat{\boldsymbol{g}}_{\textbf{H}^{(i)}_b}$ is bounded by $\gamma$, the estimation error under RR estimator for Eq.~\eqref{lineq}, is bounded by
\begin{align}
\mathbb{E}_{\text{Ridge}}
\leq \frac{\sigma_t^4 \gamma^2 + \sigma_t^2 d_{\interout^{(i)}} B^3\mu_{\max}}{(B^2\mu_{\min} + \sigma_t^2)^2}
~~, \label{error}
\end{align}
where $\sigma_t = \sigma \gamma \text{sens}_{\kappa,\text{b}}(\boldsymbol{\Omega})\|\boldsymbol{\Omega}^{-1}[t,:]\|^2_{\ell_2}$, $\mu_{\min}$ and $\mu_{\max}$ are the smallest and largest eigenvalues of ${\interout}^T_{{\boldsymbol{\phi}}^{L}_i}{\interout}_{{\boldsymbol{\phi}}^{L}_i}$, respectively. 
The detailed derivation for Eq.~(\ref{error}) and further discussion are in \appref{error}.

\begin{protofig}{Protocol $\PSecond$}{Training a Global model and Local models under Banded Matrix Factorization Mechanism}{psecond}

\textbf{Parameters:} Clients $\pclients = E_1, \ldots, E_N$, aggregation servers $\pservers = \{P_k\}^K_{k = 1}$. \\
\textbf{Preprocessing:} The servers run $\PGS$ with input $(T_\emph{num}, n_{\text{con}}, \sigma \gamma \text{sens}_{\kappa,\text{b}}(\boldsymbol{\Omega}))$ to receive $\ash{\boldsymbol{\mathcal{N}}_{\text{Tab}} }$. The servers compute $\ash{\boldsymbol{\mathcal{N}}_{\text{Corr}}}=\boldsymbol{\Omega}^{-1}\ash{\boldsymbol{\mathcal{N}}_{\text{Tab}}}$. \\

\textbf{Client-Server Phase:}~In each epoch, for each step, each client $E_i$ does the following:
\begin{enumerate}
    \item Proceeds the same steps as described in \figref{pfirst1}, but only computes and sends the $b$-th batch's $\ash{\textbf{H}^{(i)}_b}$ and $\ash{\textbf{y}_b}$ to $\pservers$. 
    \item Locally computes $\{{\partial \interout^{(i)}_{b, j}}/{{\partial \boldsymbol{\phi}}^{L}_i}\}$, and secret-shares $\ash{{\partial \interout^{(i)}_{b, j}}/{{\partial \boldsymbol{\phi}}^{L}_i}}$ for $j \in [1, B]$, then sends the shares to $\pservers$. \\
\end{enumerate}

\textbf{Server MPC:}~For each step, $\pservers$ do the following:
\begin{enumerate}
    \item Vertically concatenate $\ash{\textbf{H}^{(i)}_b}$ received from $\pclients$ for $i \in [1, N]$, which results in $\ash{\textbf{H}_b}$.
    \item Given $\ash{{\textbf{H}}_b}$ and $\ash{\textbf{y}_b}$, follow the same steps from 1 to 2, as described in \figref{pfirst1}, and eventually obtain $\ash{\mathcal{L}_j}$ and $\ash{\boldsymbol{g}_{\boldsymbol{\theta}, j}}$ for $j \in [1, B]$.
    \item Compute $\ash{\boldsymbol{g}_{\boldsymbol{\phi}^{L}_{i}, j}} = \ash{\frac{\partial  \mathcal{L}_j}{\partial \boldsymbol{\phi}^{L}_i}} = \ash{ \frac{\partial\mathcal{L}_j}{\partial  \interout^{(i)}_{b, j}}} \cdot \ash{ \frac{\partial \interout^{(i)}_{b, j}}{\boldsymbol{\partial \phi}^{L}_i}} $ for $j \in [1, B]$, with $i \in [1, N]$, using $\PMult$.
    \item Concatenate all $j$-th sample's gradient together to obtain $\ash{\boldsymbol{g}_{j}}$, for $j = [1, B]$.
    \item Follow the same steps from 3 to 5, as described in \figref{pfirst1}, and apply them to $\ash{\boldsymbol{g}_{j}}$, for $j \in [1, B]$, and eventually obtain $\ash{\overline{\boldsymbol{g}}}$. 
    \item Add noise on the aggregated gradients to obtain $\ash{ \tilde{\boldsymbol{g}}} = \frac{1}{B}(\ash{\overline{\boldsymbol{g}}} + \ash{\boldsymbol{\mathcal{N}}_{\Tilde{\boldsymbol{g}}} })$, where $\ash{\boldsymbol{\mathcal{N}}_{\Tilde{\boldsymbol{g}}} } = \ash{\boldsymbol{\mathcal{N}}_{\text{Corr}}}_{[t,:]}$.
    \item Update $\ash{\globalm}$ via gradient descent based on $\ash{ \tilde{\bs{g}}_{\theta}}$. Reconstruct $\Tilde{\boldsymbol{g}}_{\boldsymbol{\boldsymbol{\phi}^{L}_{i}}}$ and send it back to $E_i$, for $i \in [1, N]$. \\
\end{enumerate}

\textbf{Server-Client Phase:}~For each step, each client $E_i$ does the following: 
\begin{enumerate}
    \item Solves Eq.~\eqref{lineq} to get the estimation of the clipped per-sample gradients:~$\widehat{\boldsymbol{g}}_{\textbf{H}^{(i)}_b}$, based on $\Tilde{\boldsymbol{g}}_{\boldsymbol{\boldsymbol{\phi}^{L}_{i}}}$.
    \item Update the desired local layers through gradient descent, by applying back-propagation on $\widehat{\boldsymbol{g}}_{\textbf{H}^{(i)}_b}$. 
\end{enumerate}

\end{protofig}

\section{Security and Privacy Analysis}

{\subsection{Input Privacy}

Our protocols leak no intermediate values and all computation on the servers is performed entirely under MPC on secret-shared data.
Input privacy therefore follows directly from the provable security of the employed MPC building blocks.
Concretely, every sub-protocol used in our constructions operates in the \emph{arithmetic black-box} ($\mathcal{F}_{ABB}$) hybrid model, whose operations of multiplication, comparison, bit decomposition, truncation, and domain conversion are each UC-secure.
We only ever make black-box use of these subprotocols implemented in MP-SPDZ~\cite{CCS:Keller20}, and by the UC composition theorem~\cite{canetti2001universally}, the composed protocol is itself secure. We provide the theorem statement below:

\begin{theorem}[Input privacy]\label{thm:input_privacy_main}
Protocols $\PFirstA$, $\PFirstB$, and $\PSecond$ are secure against a static, semi-honest adversary corrupting at most one of the three servers, in the $\mathcal{F}_{ABB}$-hybrid model.
That is, the corrupted server's view during protocol execution can be efficiently simulated given only its input shares and output shares.
\end{theorem}

\begin{proof}
The details of the proof are given in~\ref{sec:input_privacy_formal} in \appref{mpc}.
\end{proof}


\subsection{Output Privacy}
\label{subsec:outputprivacy}

To analyze the privacy guarantee of the proposed protocols, we present the following theorems:

 \begin{theorem}\label{them1}
The protocol in~\figref{pfirst1} is $(\epsilon, \delta)$ differentially private~(\ref{def:dp}) with respect to each client $E_i$'s dataset $\mathcal{D}_i$, with $i \in [1, N]$, under Gaussian Mechanism~(\ref{def:gaussian}).
\end{theorem}


\begin{proof}
The details of the proof are given in~\ref{sec:proof} in \appref{dp}.
\end{proof}

\begin{theorem}\label{them2}
The protocol in~\figref{pfirst2} is $(\epsilon, \delta)$ differentially private~(\ref{def:dp}) with respect to each client $E_i$'s dataset $\mathcal{D}_i$, with $i \in [1, N]$, under Matrix Factorization Mechanism~(\ref{def:maxtri_factorization}).
\end{theorem}

\begin{proof}
The details of the proof are given in~\ref{sec:proof2} in \appref{dp}.
\end{proof}

\begin{theorem}\label{them3}
The protocol in~\figref{psecond} is $(\epsilon, \delta)$ differentially private~(\ref{def:dp}) with respect to each client $E_i$'s dataset $\mathcal{D}_i$, with $i \in [1, N]$, under Matrix Factorization mechanism~(\ref{def:maxtri_factorization}).
\end{theorem}

\begin{proof}
The details of the proof are given in~\ref{sec:proof3} in \appref{dp}.
\end{proof}

\subsection{Additional Discussion}
\paragraph{Label Privacy.} 
In our setup specified in Section~\ref{sec:setup}, the client storing the label information is the $N$-th client $E_N$. This label client is treated identically to the other clients in our protocols. Labels remain protected because they are never revealed to any server in plaintext, and only operated upon through secret shares, under the semi-honest, honest majority assumption. The DP guarantees in Section~\ref{subsec:outputprivacy} apply to the label party.
DP’s post‑processing property prevents label recovery from noisy aggregated gradients.
    
\paragraph{Sources of Privacy.} When only the global model is trained, no messages are passed to the clients, and thus no information is leaked to the clients. The training of the global model, including gradient computation, DP-noise addition, and model update, are all computed via MPC over the secret-shared data. Only the final model will be reconstructed and published. The servers, which are honest-but-curious, will have the same view as a user of the published model, and thus can only infer as much information as the same DP-guarantee allows.


When both the global and local models are updated, each client will additionally see a privatized gradient with respect to its own batch data reconstructed from the servers at every training step. The privacy guarantee of the underlying banded matrix factorization mechanism applies to the full sequence of privatized
gradients, and hence all intermediate models are DP-protected~\cite{Choquette-Choo2023amplified}.
}

\section{Experimental Evaluation}
\label{sec:experiments}

We empirically evaluate the performance of our protocols. Specifically, we are interested in:
\begin{enumerate}
    \item
    the model utility of our protocols in various settings,
    \item 
    {guidelines for choosing between} the protocols based on their strengths and limitations, and 
    \item
    the practicality of our protocols in implementation.
\end{enumerate}
The empirical results show that our protocols achieve competitive utility while preserving privacy; {the} test accuracies significantly {surpass those of} local-DP methods and are competitive with other baselines across various privacy budgets $\varepsilon$. Both Gaussian and BMF mechanisms can achieve good performance; each has its own advantage depending on the similarity between the public dataset for pre-training and the private dataset for the actual task. Local model finetuning offers consistent performance improvement when the learner can afford more resources for computation and communication. Our benchmarks for MPC running time and communication costs show that our protocols are practically efficient.



\subsection{Experimental Setup}

\paragraph{Plaintext Setup} The plaintext experiments were conducted on a Dell PowerEdge XE8545 rack server with 256 CPU cores and 1 TB RAM, with an Nvidia A100 GPU.
 
\paragraph{MPC Setup} MPC protocols were run on a MacBook Pro with an Apple Silicon M4 Pro processor, and 48 GB RAM on a single thread. The network connections were simulated using packet filtering (pfctl and dnctl). More specifically, we benchmark the online phases of our protocols in two setups, LAN, and WAN. On LAN, we simulate a bandwidth of 1 Gbps and a delay of 1ms, whereas on WAN we use a bandwidth of 200 Mbps and a delay of 20ms.

We implement our protocols in the MP-SPDZ~\cite{CCS:Keller20} framework. Specifically, we use the replicated-ring protocol with three parties and passive security. 
The required correlated randomness, such as edaBits, was generated in advance.






\paragraph{Framework} We consider the FL setting where the dataset is vertically split between 2 clients~($N$ = 2). 
Each client has a local ResNet-18~\cite{he2016residual} model. 
We consider \emph{two} scenarios, \textbf{Pre-ImageNet} and \textbf{Pre-Cifar100}, based on how the local models are pretrained:

In \textbf{Pre-ImageNet}, each client loads the ResNet-18 that was pre-trained on the ImageNet dataset~\cite{jia2009imagnet}. This model is publicly available in~\cite{resnet18}. The output size of the pre-trained model is 1000, i.e., $d_{\interout^{(i)}} = 1000$ for $i \in [1, N]$. 

In \textbf{Pre-Cifar100}, the local model is pre-trained on the CIFAR-100 dataset~\cite{krizhevsky2009learning} which is treated as the public dataset in our case, under the \textbf{plaintext model}. We set the output size of the pre-trained model to 512, i.e., $d_{\interout^{(i)}} = 512$ for $i \in [1, N]$.
Enhanced by data augmentation methods including data normalization and {random} flipping {during} training, the pre-trained model achieves an average testing accuracy of $77.19\%$.
On the servers' end, we deploy one fully-connected layer as the global model.

We also note that using deeper networks or advanced techniques, such as Sharpness-Aware Minimization (SAM), can lead to significant improvements in accuracy, as shown in~\cite{foret2021sharpness}. However, our goal is not to employ state-of-the-art methods to maximize performance. Instead, we adopt a standard training approach that yields moderate yet reasonable results. This simplicity makes our setting a practical choice for real-world deployment scenarios.


\paragraph{Datasets} 
{We evaluate our framework on CIFAR-10~\cite{krizhevsky2009learning} and EMNIST~\cite{cohen2017emnist}, which are treated as private datasets in our experiments. As mentioned above, CIFAR-100 dataset is the public dataset used for pre-training. Each data sample in all datasets is vertically partitioned based on the number of participating clients, and the corresponding feature subset is subsequently forwarded to the respective client’s local model during forward-passing. 

For EMNIST, we split the dataset into 26 classes by letters. We assume that the data distribution across all clients is independently and identically distributed~(IID). As DP mechanisms provide distribution-independent theoretical privacy guarantees, the IID or Non-IID setting under our framework \emph{does not} affect the formal DP protection. We acknowledge that heterogeneous data distributions may influence model utility or empirical attack performance~\cite{khan2025review,madabushi2025opus}. However, a detailed study of their impact on practical privacy leakage is beyond the scope of this work.  } 

\paragraph{Protocol Setups} 
We use $\textbf{Plain}$ to denote the plaintext protocol in \secref{plian}, $\textbf{G-Shuff}$ to represent the protocol described in~\secref{mpc_dp_localfreez_shuffle}, $\textbf{G-BMF}$ to denote the protocol in~\secref{mpc_dp_localfreez_bmp}, and $\textbf{GL-BMF}$ to denote the protocol introduced in \secref{mpc_dp_gl_update}.



In \textbf{GL-BMF}, each client solves the linear equations~Eq.~\eqref{lineq} through {the} RR estimator with~$\lambda = (\sigma \gamma \text{sens}_{\kappa,\text{b}}(\boldsymbol{\Omega})\|\boldsymbol{\Omega}^{-1}[t,:]\|^2_{\ell_2}/B)^2$, at each step.
Besides, we apply LoRA finetuning~\cite{hu2022lora} to clients' local models. Instead of modifying all model parameters of the local model, we insert LoRA layers into the local model architecture as the \emph{adapter} to the new task. More specifically, we insert one LoRA layer to each block of the local ResNet model and only the LoRA layers are updated during training.
We also organize the LoRA layers in \emph{parallel}: the output of each LoRA layer is directly added to the final output of the local model instead of the next ResNet block. This choice is inspired by the finding in~\cite{lin2017fpn,rogers2020primer}:~Intermediate features can capture rich hierarchical information, hence incorporating LoRA layers in a parallel structure enables efficient adaptation across multiple feature extraction layers. 
Moreover, since the estimated per-sample gradients $\widehat{\boldsymbol{g}}_{\textbf{H}^{(i)}_b}$ contain noise, our parallel structure, which allows each layer to be updated independently, may mitigate the accumulation of {noise} as opposed to the conventional \emph{sequential} LoRA structure where {noise propagates} across layers and {is} particularly amplified on the bottom layers during back-propagation.  

In addition, under our parallel LoRA structure, we have 
\begin{align}
\interout^{(i)}_{b, j} = \interout^{(i),o}_{b, j} + \interout^{(i),L_1}_{b, j} +\dots +  \interout^{(i),L_M}_{b, j}~~, \nonumber
\end{align}
where $\interout^{(i),o}_{b, j} $ is the output of the original model and $\interout^{(i),L_l}_{b, j} $ is the output of the $l$-th LoRA layer, {where} $M$ is the total number of LoRA layers. Based on the principles of differentiation, we have
\[
\frac{\partial \mathcal{L}_j}{\partial \interout^{(i)}_{b, j} } = \frac{\partial \mathcal{L}_j}{\partial \interout^{(i), L_1}_{b, j} }=\dots =\frac{\partial \mathcal{L}_j}{\partial \interout^{(i), L_M}_{b, j} }~~.
\]
Thus, instead of uploading all $M$ layers' ${\partial \interout^{(i), L_l}_{b, j}}/{{\partial \boldsymbol{\phi}}^{L_l}_i}$ with $j \in [1, B]$ where $\boldsymbol{\phi}^{L}_i = \{\boldsymbol{\phi}^{L_l}_i\}^{M}_{l=1}$ to servers for computing $\{\Tilde{\boldsymbol{g}}_{\boldsymbol{\phi}^{L_l}_{i}}\}^{M}_{l=1}$, each client could just upload any one layer's ${\partial \interout^{(i), L_l}_{b, j}}/{{\partial \boldsymbol{\phi}}^{L_l}_i}$ to the servers to compute $\Tilde{\boldsymbol{g}}_{\boldsymbol{\phi}^{L_l}_{i}}$, and hence obtain the estimation $\widehat{\boldsymbol{g}}_{\textbf{H}^{(i), L_l}_b} $. We then approximate $\widehat{\boldsymbol{g}}_{\textbf{H}^{(i)}_b} \approx \widehat{\boldsymbol{g}}_{\textbf{H}^{(i), L_l}_b}$, and update all LoRA layers by $\widehat{\boldsymbol{g}}_{\textbf{H}^{(i)}_b} $. Here we set $r= 8$ and ${\alpha}_{\text{LoRA}} = 1$ for all LoRA layers across all clients in our experiments. 


{\paragraph{Parameters}~In our experiments, we use cross-entropy as the criterion and SGD as the optimization algorithm. We fix $B = 128$, and set $E_\emph{num} $ to be 20 and 100, respectively. For DP testing, we set $\ell_2$ norm bound $\gamma = 1.2$, $\delta = 10^{-5}$, and various privacy budgets with $\epsilon = 1, 2, 5, 8 , 10$. For the BMF mechanism, we consider two sets of learning parameters:
\begin{itemize}
    \item \textbf{Setting 1}:~$(\alpha = 1, \beta = 0)$
    \item \textbf{Setting 2}:~$(\alpha = 1, \beta = 0.9)$
\end{itemize}
Besides, we set learning rate $\eta_s = \eta_i = 0.01$ on both server and clients for $\textbf{Plain}$, and $\eta_s = 0.01$ on servers for $\textbf{G-Shuff}$ and $\textbf{G-BMF}$. For $\textbf{GL-BMF}$, we set $\eta_s = 0.01 $ and $\eta_i = 0.001$ for $i \in [1, N]$. Here we use smaller $\eta_i $ to mitigate the impact of the estimation error contained in the estimated per-sample gradients. All the reported test accuracies are averaged over \textbf{four} independent runs. 

\paragraph{DP Accounting}~For DP accounting, for the subsampling-based protocol, we follow the DP-accounting and composition summarized in~\cite{abadi2016privacy,Mironov2017renyi} and implemented using DP library Opacus~\cite{opacus}; For the BMF-based protocols, we use Google’s FFT based DP accounting library~\cite{gdpl}, following the exact setup in~\cite{Kalinin2024banded}.
}
\paragraph{Local Differential Privacy~(LDP)}
In the VFL setting with LDP, each individual perturbs their own output by adding noise, without relying on a trusted party. To provide a direct comparison with our method, we adopt LDP as one of the baselines. Same as the work in~\cite{wu2024federated}, we have each client independently {privatize} its data before transmitting to the server. Specifically, each client applies norm clipping to its per-sample intermediate outputs $\interout^{(i)}_{b, j}$ for $j \in [1, B]$ to bound sensitivity. A calibrated Gaussian noise is then added to each clipped output. As a result, each party releases only perturbed intermediate outputs to the servers. However, unlike~\cite{wu2024federated}, which utilizes subsampling enabled by fuzzy linkage (one-to-many linkage), we adopt the conventional VFL setting where each primary sample has exactly one predetermined match in other parties. Hence the sample-level participation anonymity among clients is no longer preserved after data alignment, and privacy amplification that relies on subsampling is also no longer applicable. We use $\textbf{LDP-G}$ and $\textbf{LDP-GL}$ to denote the mechanisms for updating global model only and global-local models both under LDP, respectively. 
{\paragraph{Existing Approaches}~We further compare our protocols against existing privacy-preserving VFL approaches, which can be viewed as \emph{variants} of LDP under \emph{different} trust assumptions, used as baselines. Specifically, we consider \textbf{ADMM}~\cite{xie2024admm}, \textbf{Split Learning}~\cite{vepakomma2018splitlearning}, and \textbf{FedBCD}~\cite{liu2022fedbcd} as in~\cite{xie2024admm}, along with \textbf{Ada-VFed}~\cite{gai2025differentially}. To ensure a fair comparison, we adopt the same experimental configuration as \textbf{Ada-VFed}:~In particular, two clients each deploy a ResNet-18 model locally, and the server employs an MLP layer as the global model. The clipping threshold is set to $\gamma = 2$, the number of epochs $E_{num} = 20$, the batch size $B = 128$, and various privacy budgets with $\epsilon = 1, 4, 8$. All methods are evaluated on the CIFAR-10 dataset. For \textbf{ADMM} and \textbf{FedBCD}, we evaluate multiple numbers of local update steps, including $\{1, 5, 10, 30\}$, and fix the value to 5, as it achieves the best performance among these settings. Similarly, after extensive hyperparameter tuning, we fix $\rho = 0.5$ for \textbf{ADMM}.

\paragraph{Membership Inference Attacks~(MIAs)}~MIAs are privacy attacks which aim {to determine} whether a target data point $(x,y)$, was included in the given model’s training set. As described in Section~\ref{sec:security_definition}, we assess the robustness of our privacy-preserving protocols against these adversaries by conducting black-box membership inference attacks with the Privacy Meter introduced in~\cite{mlprivacymeter}. Specifically, we use the enhanced Population Attack~(\textbf{P-Attack}) proposed in~\cite{ye2022enhanced} to evaluate the privacy leakage. We adopt the same attack setting as in~\cite{ye2022enhanced,bao2025deep}, where the validation dataset which is never used by the model during training, is used to construct the population distribution and compute the attack threshold. When evaluating attacks, we keep using {a} balanced evaluation dataset:~$50\%$ member~(from training dataset) and $50\%$ non-member~(from test dataset).
}
\subsection{Benchmarks}
\begin{table*}[!htbp]
\centering
{\small
\begin{tabular}{cc|ccccc|ccccc}
\toprule
\multirow{1}{*}{\textbf{Dataset}}    & \multirow{1}{*}{\textbf{Method}}          
&\multicolumn{5}{c|}{\textbf{Pre-ImageNet}} &\multicolumn{5}{c}{\textbf{Pre-Cifar100}} \\
\midrule
\multirow{9}{*}{CIFAR-10}  & &\multicolumn{5}{c|}{\textbf{Non-DP}} &\multicolumn{5}{c}{\textbf{Non-DP}} \\



  & $\text{Plain}$   &\multicolumn{5}{c|}{90.01 } &\multicolumn{5}{c}{91.42}  \\
 
&  &\multicolumn{5}{c|}{\textbf{DP}} &\multicolumn{5}{c}{\textbf{DP}} \\
& & ${\varepsilon=1}$            & ${\varepsilon=2}$           & ${\varepsilon=5}$ & ${\varepsilon=8}$   & ${\varepsilon=10}$     & ${\varepsilon=1}$            & ${\varepsilon=2}$           & ${\varepsilon=5}$    & ${\varepsilon=8}$     & ${\varepsilon=10}$  \\

 & $\text{G-Shuff}$  & 69.38 & 71.78 & 73.27 & 73.69 & 73.76 &81.19 & 81.51 & 81.62 & 81.67 & 81.69 \\
  & $\text{G-BMF}$~(Setting~1)   & 69.25 & 70.73 & 71.19 & 71.22 & 71.24 &77.22& 77.28 & 77.31& 77.34 & 77.37\\
  & $\text{G-BMF}$~(Setting~2)   & 51.74 & 58.69 & 67.29 & 70.27 & 71.59 & 72.03 & 77.39 & 80.68 & 81.12 & 81.38 \\
& $\text{GL-BMF}$~(Setting~1)  & 69.63 & 71.09 & 71.48 & 71.53 &71.56 & 77.48 & 77.55 & 77.77 & 77.81 & 77.82 \\
& $\text{GL-BMF}$~(Setting~2)  & 54.41 & 60.45 & 68.19 & 71.01 &72.34 & 72.58 & 77.86 & 81.11 & 81.39 & 81.42 \\
\midrule
\multirow{9}{*}{EMNIST}  & &\multicolumn{5}{c|}{\textbf{Non-DP}} &\multicolumn{5}{c}{\textbf{Non-DP}} \\

  & $\text{Plain}$   &\multicolumn{5}{c|}{94.12} &\multicolumn{5}{c}{93.75}  \\
 
&  &\multicolumn{5}{c|}{\textbf{DP}} &\multicolumn{5}{c}{\textbf{DP}} \\
& & ${\varepsilon=1}$            & ${\varepsilon=2}$           & ${\varepsilon=5}$   & ${\varepsilon=8}$    & ${\varepsilon=10}$    & ${\varepsilon=1}$            & ${\varepsilon=2}$           & ${\varepsilon=5}$   & ${\varepsilon=8}$    & ${\varepsilon=10}$      \\

 & $\text{G-Shuff}$ & 73.96 & 76.71 & 78.37 & 78.94 & 79.22 &78.58 & 78.87 & 79.02 &79.12 & 79.14\\
& $\text{G-BMF}$~(Setting~1)    & 73.62 & 74.53 & 74.89 & 75.03 &75.05  & 70.88 & 70.94& 70.99 & 71.04 & 71.09  \\
& $\text{G-BMF}$~(Setting~2)    & 73.44 & 78.95 & 81.48 & 81.99 &82.06 & 78.69 & 79.73 & 80.08 & 80.09 & 80.13 \\
& $\text{GL-BMF}$~(Setting~1)   & 74.09 & 74.84 & 75.12 & 75.25 & 75.28 & 71.06 & 71.15 & 71.19 & 71.21 & 71.23  \\
& $\text{GL-BMF}$~(Setting~2)   & 74.04 & 79.52 & 81.84 & 82.18 &82.27 & 78.99 & 79.94 & 80.19 & 80.21 & 80.24  \\
\bottomrule
\end{tabular}
}
\caption{Test Accuracy~($\%$) versus Different Datasets, Pre-training, and Privacy Budgets~(20 epochs)}
\label{tab1}
\end{table*}
\begin{table*}[!htbp]
\centering
{\small
\begin{tabular}{cc|ccccccc}
\toprule
\multirow{2}{*}{\textbf{Dataset}} & \multirow{2}{*}{$\epsilon$}          
&\multicolumn{7}{c}{\textbf{Method}} \\
 & & {\text{G-Shuff}} & $\text{G-BMF}$~(Setting~1)  & $\text{G-BMF}$~(Setting~2)  & $\text{GL-BMF}$~(Setting~1) &  $\text{GL-BMF}$~(Setting~2) &  $\text{LDP-G}$  &  $\text{LDP-GL}$\\
 \midrule
\multirow{5}{*}{CIFAR-10} &  $1$ & 63.71 & 65.14 & 49.14 & 65.97 & 51.32 & 10.29 & 10.56\\
&  $2$ & 69.09 & 70.66 & 55.16 & 71.23 & 56.26 & 11.14 & 11.36\\
&  $5$ & 71.93 & 73.63 & 61.47 & 74.09 & 62.29 & 16.69 & 17.14\\
&  $8$ & 72.89 & 74.25 & 64.42 & 74.59& 65.32 & 23.97 & 24.51\\
& $10$ & 73.88 & 74.41 & 65.85 & 74.78 & 66.61 & 28.16 & 28.89\\
\bottomrule
\end{tabular}
}
\caption{Test Accuracy~($\%$) versus Different Privacy Budgets, LDP~(100 epochs)}
\label{ldp}
\end{table*}
\begin{table*}[!htbp]
\centering
{\small
\begin{tabular}{cc|ccccccc}
\toprule
\multirow{2}{*}{\textbf{Dataset}} & \multirow{2}{*}{$\epsilon$}          
&\multicolumn{7}{c}{\textbf{Method}} \\
 & & {\text{G-Shuff}} & $\text{G-BMF}$  & $\text{GL-BMF}$ & $\text{ADMM}$ &  $\text{FedBCD}$ &  $\text{Split Learning}$  &  $\text{Ada-VFed}$\\
 \midrule
\multirow{3}{*}{CIFAR-10} &  $1$ & 66.71 & 67.14 & 68.49 & 65.79 & 68.22 & 29.66 & 61.31\\
&  $4$ & 71.25 & 70.96 & 72.52 & 68.84 & 70.43 & 34.41 & 62.13\\
&  $8$ & 72.37 & 71.91 & 73.63 & 70.03 & 71.68 & 40.37 & 63.50\\
\bottomrule
\end{tabular}
}
\caption{{Test Accuracy~($\%$) versus Different Privacy Budgets, Baselines~(20 epochs, $\gamma = 2$)}}
\label{baselines}
\end{table*}
\begin{table*}[!htbp]
\centering
{\small
\begin{tabular}{c|cccccc}
\toprule
\multirow{2}{*}{\textbf{Dataset}}          
&\multicolumn{6}{c}{\textbf{Method}} \\
 & {\text{Plain}}  & {\text{G-Shuff}} & $\text{G-BMF}$~(Setting~1)  & $\text{G-BMF}$~(Setting~2)  & $\text{GL-BMF}$~(Setting~1) &  $\text{GL-BMF}$~(Setting~2) \\
 \midrule
\multirow{1}{*}{CIFAR-10}  & 0.68 & 0.52 & 0.51 & 0.51 & 0.51 & 0.51 \\
\bottomrule
\end{tabular}
}
\caption{{AUC Score versus Different Protocols under \text{P-Attack}~(100 epochs, $\epsilon = 8$)}}
\label{auc}
\end{table*}
\begin{table*}[!htbp]
\centering
{\small
\begin{tabular}{cc|cc|cc}
\toprule
\multirow{2}{*}{\textbf{Dataset}}    & \multirow{2}{*}{\textbf{Method}}          
&\multicolumn{2}{c|}{\textbf{Pre-ImageNet}} &\multicolumn{2}{c}{\textbf{Pre-Cifar100}} \\
 & & {\textbf{Runtime}~(s)} & {\textbf{Communication}~(MB)} & {\textbf{Runtime}~(s)} &{\textbf{Communication}~(MB)} \\
\midrule
\multirow{4}{*}{CIFAR-10}
&  $\text{G-Shuff}$ & 14.88 & 473.63 & 7.80 & 244.03 \\
&  $\text{G-BMF}$ & 14.57 & 453.6 & 7.65 & 233.77 \\
& $\text{GL-BMF}$ & 132.8 & 6220.87 & 83.5 & 3186.62 \\
 \midrule
\multirow{4}{*}{EMNIST}
 & $\text{G-Shuff}$ & 32.80 & 1228.56 & 18.89 & 632.52 \\
 & $\text{G-BMF}$ & 32.00 & 1178.64 & 18.50 & 606.96 \\
 & $\text{GL-BMF}$ & 133.87 & 6945.87 & 89.33 & 3559.79 \\
\bottomrule
\end{tabular}
}
\caption{Runtime and Communication cost of the MPC over LAN}
\label{tab:mpclan}
\end{table*}
\paragraph{Utility-Privacy Evaluation.} Tab.~\ref{tab1} shows the utilities for different protocols under different privacy budgets after 20 training epochs. We observe that all our protocols yield good utilities for both scenarios and both datasets. In the \textbf{Pre-Cifar100} scenario, our protocols perform exceptionally well as the difference {in} final model accuracies between the plaintext upperbound and our privacy-preserving protocols is less than 15\% across all $\varepsilon$. The performance of BMF mechanisms can depend on the choice of hyperparameters in learning. We consider the better of Setting 1 and 2 as it shows the potential of the mechanism with the right choice of hyperparameters. The \textbf{Pre-ImageNet} scenario is more challenging because the images used for pretraining differ more from the test data than \textbf{Pre-Cifar100} in terms of image size and contents.
Nevertheless, our protocols still achieve competitive performances of less than 20\% accuracy gap in most cases except for small $\varepsilon=1,2$.

The second general pattern is that \textbf{GL-BMF} consistently outperforms \textbf{G-BMF} across all settings. This pattern indicates that local finetuning helps. {Although} the local models are only updated with \emph{noisy} gradients estimated from the published privacy-preserving information from the servers, the extracted information is still useful. Although the improvement is moderate because of the large amount of noise added in {the} BMF mechanism, we see potential in the idea of decoding local gradients from privacy-preserving public information. With the potential development of noise-adding mechanisms that consider our decoding step in the design, we may benefit further from local finetuning.  

For the choice between \textbf{G-Shuff} and \textbf{G-BMF}, we notice that $\textbf{G-BMF}$, with the right choice of hyperparameters, generally outperforms \textbf{G-Shuff}. However, \textbf{G-Shuff} outperforms other baselines in \textbf{Pre-Cifar100} -- \textbf{CIFAR-10}. The relative performance between \textbf{G-Shuff} and \textbf{G-BMF} is related to the similarity between the dataset for pre-training and the test dataset; the phenomenon is an intriguing direction for further exploration. The BMF mechanism supports local finetuning, which offers more potential room for improvement. On the other hand, the performance of \textbf{G-BMF} is sensitive to the choice of hyperparameters in learning. (See the difference between Setting 1 and 2.) One should also consider the privacy budget for model selection in more complex learning settings. 

In addition, we trace the performance of all our protocols with a longer training horizon ($E_{num} = 100$). The results are shown in Tab.~\ref{tab1_1} and the detailed discussion is given in \appref{add}.

We also compare our protocols with Local-DP under $E_{num} =100$. Since there is no prior work with end-to-end input/output privacy guarantees for general VFL, we use local-DP which can be regarded as the variant of~\cite{wu2024federated} without subsampling, as a privacy-preserving benchmark for utility. Tab.~\ref{ldp} shows the utilities for both our protocols and LDP for \textbf{CIFAR-10} in the \textbf{Pre-ImageNet} scenario. After 100 {epochs} of training, the performances of both \textbf{LDP-G} and \textbf{LDP-GL} are underwhelming. The LDP protocols add noise before any data and/or statistics leave a client's local storage. The {noise, which is} added early in the pipeline, {causes} difficulties for all downstream calculations. In contrast, our combination of MPC and DP is much more privacy-budget friendly as the privacy-preserving {noise is} added {to} aggregated statistics later in the pipeline. The significant gap between our protocols and Local-DP shows the advantage of our approach in privacy-utility trade-off.

{In Tab.~\ref{baselines}, we compare our protocols with other baselines described in Section~\ref{sec:experiments}. All the approaches are evaluated on the CIFAR-10 dataset under \textbf{Pre-ImageNet} scenario. For approaches from~\cite{xie2024admm}, we set $d_{\interout^{(i)}} = 60$ to remain consistent with their original configuration. Here both \textbf{G-BMF} and \textbf{{GL-BMF}} are under Setting~1. In general, our protocols consistently outperform \textbf{Ada-VFed} and \textbf{Split Learning}:~For example, at $\epsilon = 8$, they achieve a minimum of 8\% higher performance than \textbf{Ada-VFed}, and an even greater improvement over \textbf{SplitLn}. This result indicates that, on complex datasets such as CIFAR-10, the Laplacian Score metric used in \textbf{Ada-VFed} fails to capture meaningful feature importance, resulting in performance degradation. We also observe that all our protocols outperform \textbf{ADMM} across all privacy budget settings. For \textbf{FedBCD}, under the strict privacy constraint ($\epsilon = 1$), it outperforms \textbf{G-shuff} and \textbf{G-BMF} and is competitive with \textbf{GL-BMF}. However, when under the relaxed privacy constraints ($\epsilon \geq 4$), all our protocols perform better than \textbf{FedBCD}, with \textbf{GL-BMF} showing particular improvement. 

Furthermore, we find that \textbf{FedBCD} performs better than \textbf{ADMM} in our experiments, contrary to the reported results in~\cite{xie2024admm}. 
\textbf{ADMM}'s sophisticated consensus mechanism may be sensitive to changes in the experiment configuration, such as the size of clients and local model architecture.  
This degradation can also be explained by \textbf{ADMM}'s sensitivity to hyperparameters:~The penalty parameter $\rho$ and dual variable updates must be extensively re-tuned for each different setting, which limits its practical applicability for general usage. In contrast, \textbf{FedBCD} is more robust, as it has fewer hyperparameters to tune and its block coordinate descent enables efficient small-scale coordination. 

The extended comparisons with additional baseline methods and further discussion are provided in~\appref{add}.

Tab.~\ref{auc} reports the effectiveness of the population-based membership inference attack (P-Attack) against our protocols compared to the \textbf{Plain} protocol, for \textbf{CIFAR-10} in the \textbf{Pre-ImageNet} scenario. All evaluations are conducted under a privacy budget of $\epsilon = 8$, representing a relatively loose privacy constraint. Under this setting, the plaintext model reaches an AUC of 0.68, whereas all of our protocols keep the AUC at 0.52 or below, demonstrating strong resistance to the membership inference attack.
}

\paragraph{MPC Running Evaluation Results} Tab.~\ref{tab:mpclan} shows the MPC runtime and communication cost for different protocols over LAN. Note the results presented in both Tab.~\ref{tab:mpclan} and Tab.~\ref{tab:mpcwan}~(see in~\appref{add}) are for a single batch. For simulating the runtime of \textbf{G-Shuff} at the batch-level, we first measure the runtime of shuffling on $\ash{\overline{\textbf{H}}}$, and then divide the runtime by $B_{num}$. The shuffling runtime at the batch-level is then added {to} the time spent in the {rest} of the protocol for a single batch. To the best of our knowledge, there have not been any works performing training of ResNet-18 completely in MPC. Previous works have only performed training on neural networks such as AlexNet~\cite{EPRINT:KelSun22}, which is a much smaller network. Adapting these frameworks to support the complexities of ResNets such as skip connections is non-trivial. 

Unsurprisingly, both the \textbf{G-Shuff} and \textbf{G-BMF} protocols are much faster than \textbf{GL-BMF}. This is due to the $B\times N$ ($B$ per client) additional matrix multiplications needed in the \textbf{GL-BMF} protocol, for computing $ \ash{\boldsymbol{g}_{\boldsymbol{\phi}^{L}_{i}, j}}$ with $j \in [1, B]$. In general, the online times and communications for the \textbf{G-Shuff} and \textbf{G-BMF} are very close, {as the primary difference between the two is} 
due to the extra shuffling step conducted at each epoch under the \textbf{G-Shuff} variant. But we should also note after accumulating the per-batch runtime to per-epoch runtime, {on} average, \textbf{G-Shuff} is still more time-consuming than \textbf{G-BMF}:~Such as when under \textbf{Pre-ImageNet} \text{--} \textbf{CIFAR-10}, the difference of the runtime at epoch level gets to around $97$ seconds, which is non-negligible. Moreover, all observations {summarized in} Tab.~\ref{tab:mpclan} also appear in Tab.~\ref{tab:mpcwan}. 

\section{Conclusion}

In conclusion, we propose a novel end-to-end privacy-preserving framework,
instantiated by three efficient protocols for different deployment
scenarios, addressing both input and output privacy, for VFL. Our method avoids the high computation and communication costs of naive MPC-based solutions while maintaining model utility, making it practical for real-world deployments. We also observe that the improvement of our global-local model update protocol over the pure global model update protocol is limited. This is because the strong noise introduced by BandMF at each step significantly impairs the linear estimation, which in turn degrades the final model performance. This observation motivates another direction for future work:~By developing more efficient DP mechanisms that enhance utility, our PP-VFL protocol could achieve substantial performance gains.  


\begin{acks}
This research received no specific grant from any funding agency in the public, commercial, or not-for-profit sectors. The authors used generative AI-based tools to revise the text, improve flow and correct any typos, grammatical errors, and awkward phrasing.
\end{acks}

\bibliographystyle{ACM-Reference-Format}
\bibliography{main.bib,abbrev3.bib,shrunk.bib}

\appendix
\section{Other Related Work}\label{app:related_work}

\textbf{Secure Aggregation and PPML via MPC.}~A core strand of research in the privacy-preserving machine learning (PPML) community leverages secure multi-party computation (MPC) to enable collaborative model training while preserving input privacy across different settings. Several works have developed MPC-based frameworks to address practical PPML challenges, including floating-point arithmetic, efficient dot-product computation, and non-linear activation functions
~\cite{demmler2015aby,payman2018aby3,CCS:ABFKLO18,tetrad,C:EGKRS20}. However, these techniques are not specifically tailored for FL or VFL, where data partitions and interaction patterns differ significantly from conventional PPML scenarios.

\textbf{MPC \& DP.}~A smaller subset of the community has explored combining MPC protocols with differential privacy techniques, aiming to achieve complementary advantages of both: secure protocol execution and strong formal privacy bounds~\cite{CCS:BohKer21,CCS:ChasheUll19,FC:EIKN21,das2025communication}. However, to our knowledge, these combined approaches have not been systematically studied in the setting of vertical federated learning.

\textbf{FL \& MPC.}~In the setting of federated learning, MPC has been used to enable protocols for secure aggregation~\cite{Bonawitz2017practical}, which focuses on protecting individual local models' updates within FL. Protocols such as \cite{CCS:BBGLR20,SP:MWAPR23,SP:LBVKH23,CCS:CGJv22} focus on the single-server setting, while \cite{SP:RSWP23,scionfl} {consider} the multi-server setting. 
While these works have demonstrated promising performance, they predominantly consider horizontally partitioned data (i.e., each party holds the same set of features but on different subsets of samples), and it is non-trivial to extend them to the vertically split scenario. Recent work in~\cite{qiu2023efficient,qiu2024secure} {proposes} a secure {layer-based} framework for training vertical FL securely through secure aggregation. Although this framework can protect the private data, the output privacy of the trained {model} is still not guaranteed, which makes the framework {vulnerable} to attacks such as MIAs. 

\textbf{FL \& HE.}~Another line of research employs encryption-based techniques, particularly homomorphic encryption (HE), to preserve input privacy during model training. In~\cite{Froelicher2023scalable}, the authors propose a federated Principal Component Analysis (PCA) framework over distributed datasets using HE, but their method is limited to the horizontal FL setting and also doesn't consider the output privacy-preserving. In contrast, the works in~\cite{fate} consider a vertical FL-style setting and applies HE to protect both the upstream intermediate representations and also downstream gradients. In~\cite{xu2021fedv}, the authors exploits the functional encryption fore secure computation under VFL scenario. However, both of these approaches focused only on the input privacy and does not address the output privacy issues.

\textbf{DP in Horizontal FL.}~Previous works incorporate DP into horizontal federated learning~(HFL)~\cite{KLS22,fldp1}. These methods typically add calibrated noise to gradients during training, ensuring that individual client contributions remain obscured. Recent studies~\cite{mal2024noise,tim2022efficient} further focus on enhancing both the utility and scalability of output {privacy preservation} in HFL, going beyond conventional DP mechanisms. Such DP-based FL solutions are practical for deployment and have been extensively analyzed in terms of both privacy and utility. However, they are predominantly tailored to the HFL setting and cannot be directly extended to {the} VFL scenario. 

\section{Additional Details and Experimental Results}\label{app:add}
This section provides additional materials that complement the main manuscript, including the notation table~(Table~\ref{tab:notation}) used throughout the paper and an extended experimental results on communication and computation costs of the MPC over WAN~(Tab.~\ref{tab:mpcwan}). Besides, additional experimental results include:

\paragraph{Results under $E_{num}=100$}~Tab.~\ref{tab1_1} reports the utilities for different protocols after 100 training epoches. The general patterns are consistent. The BMF mechanism seems to benefit more from increasing training epoches. For example, in \textbf{Pre-ImageNet} -- \textbf{EMNIST}, the performance of \textbf{G-BMF} increases from 73.62 to 75.28 at $\varepsilon=1$, whereas the performance of \textbf{G-Shuff} does not improve. Recall that BMF adds correlated noises with the hope that the noises can "offset" among themselves over multiple epoches. Longer training horizon potentially benefit BMF. However, we must also note that the longer training horizon may favor a different choice of hyperparameters. For example, in \textbf{Pre-ImageNet} -- \textbf{EMNIST} at $\varepsilon = 2$, \textbf{G-BMF} used to perform better with the hyperparameters in Setting 1, but now achieves higher model accuracy with the hyperparameters in Setting 2. This phenomenon reminds us practitioners to carefully consider the characteristics of Gaussian v.s. BMF mechanisms.

{\paragraph{Extended Baselines Comparisons}~A central design principle of our protocols is to strictly preserve sample-level anonymity throughout training. Relaxing this constraint would enable the global-local training protocol \textbf{GL} to exploit privacy amplification by subsampling. Consistent with the setting in~\cite{xie2024admm}, the servers may randomly sample a subset of $B$ data indices at each step and broadcast them to the clients. Under this modification, \textbf{GL-BMF} can employ the subsampled Gaussian mechanism in place of the BMF mechanism. We denote this variant of \textbf{GL-BMF} as \textbf{GL-Sub}. We stress that this variant is \textbf{not} proposed as a new formal privacy-preserving protocol, as it introduces significant privacy risks by breaking the anonymity of sample participation. However, we introduce this variant solely as a special baseline for comparison with our formal \textbf{GL-BMF} protocol. While \textbf{GL-Sub} relies on much different trust and privacy assumptions, it employs the same denoising techniques as we used in \textbf{GL-BMF}. 

We follow the same experimental settings as in Tab.~\ref{baselines}, but set the number of training epochs to $E_{\text{num}} = 100$ and $\epsilon = 4$. We compare \textbf{GL-BMF} (under Setting 1) with \textbf{GL-Sub}, as well as \textbf{ADMM} and \textbf{FedBCD}. The results in Tab.~\ref{extend} show that \textbf{GL-BMF} continues to outperform both \textbf{ADMM} and \textbf{FedBCD} over a longer training horizon. However, we can also clearly find that by utilizing privacy-amplification, \textbf{GL-Sub} outperform all other protocols. This is because when under weaker noise perturbation (the subsampling-based Gaussian mechanism), the denoising procedure applied to the private gradients is more effective than under the BMF mechanism, which injects stronger noise at each step, resulting in an overall boost in utility. We leave further discussion to the readers to assess which privacy condition is most relevant to them and to determine the privacy–utility trade-off that best fits their practical applications.
}
\begin{table*}[h!]
\centering
\caption{Notation Table}
\begin{tabular}{ll}
\toprule
\textbf{Symbol} & \textbf{Description} \\
\midrule
$\globalm$ & Global model \\
$\boldsymbol{\mathcal{F}}_i$ & Local model at client $i$\\ 
$\boldsymbol{\theta}$ & Parameters of the global model $\globalm$\\
$\boldsymbol{\phi}_i$ & Parameters of local model $\boldsymbol{\mathcal{F}}_i$ \\
$\boldsymbol{x}^{(i)}_{b, j}$ & Client $i$'s local data of the $j$-the sample in the $b$-th batch\\
$\interout^{(i)}_{b, j}$ & Output of local model $\boldsymbol{\mathcal{F}}_i$ over $\boldsymbol{x}^{(i)}_{b, j}$
\\
\bottomrule
\end{tabular}
\label{tab:notation}
\end{table*}
\begin{table*}[!htbp]
\centering
{\small
\begin{tabular}{cc|ccccc|ccccc}
\toprule
\multirow{1}{*}{\textbf{Dataset}}    & \multirow{1}{*}{\textbf{Method}}          
&\multicolumn{5}{c|}{\textbf{Pre-ImageNet}} &\multicolumn{5}{c}{\textbf{Pre-Cifar100}} \\
\midrule
\multirow{9}{*}{CIFAR-10}  & &\multicolumn{5}{c|}{\textbf{Non-DP}} &\multicolumn{5}{c}{\textbf{Non-DP}} \\



  & $\text{Plain}$   &\multicolumn{5}{c|}{93.5 } &\multicolumn{5}{c}{93.89 }  \\
 
&  &\multicolumn{5}{c|}{\textbf{DP}} &\multicolumn{5}{c}{\textbf{DP}} \\
& & ${\varepsilon=1}$            & ${\varepsilon=2}$           & ${\varepsilon=5}$ & ${\varepsilon=8}$   & ${\varepsilon=10}$     & ${\varepsilon=1}$            & ${\varepsilon=2}$           & ${\varepsilon=5}$    & ${\varepsilon=8}$     & ${\varepsilon=10}$  \\

 & $\text{G-Shuff}$  & 63.71 & 69.09 & 71.93 & 72.89 & 73.88 &80.16 & 82.11 & 82.85 & 82.98 & 83.02 \\
  & $\text{G-BMF}$~(Setting~1)   & 65.14 & 70.66 & 73.63 & 74.25 & 74.41 &79.51 & 80.41 & 80.57 & 80.59 & 80.64\\
  & $\text{G-BMF}$~(Setting~2)   & 49.14 & 55.16 & 61.47 & 64.42 & 65.85 & 66.83 & 71.31 & 78.18 & 80.61 & 81.49 \\
& $\text{GL-BMF}$~(Setting~1)  & 65.97 & 71.23 & 74.09 & 74.59 &74.78 & 79.79 & 80.75 & 80.81 & 80.94 & 81.05 \\
& $\text{GL-BMF}$~(Setting~2)  & 51.32 & 56.26 & 62.29 & 65.32 &66.61 & 67.44 & 71.87 & 78.66 & 80.99 & 81.83 \\
\midrule
\multirow{9}{*}{EMNIST}  & &\multicolumn{5}{c|}{\textbf{Non-DP}} &\multicolumn{5}{c}{\textbf{Non-DP}} \\

  & $\text{Plain}$   &\multicolumn{5}{c|}{94.26} &\multicolumn{5}{c}{94.51}  \\
 
&  &\multicolumn{5}{c|}{\textbf{DP}} &\multicolumn{5}{c}{\textbf{DP}} \\
& & ${\varepsilon=1}$            & ${\varepsilon=2}$           & ${\varepsilon=5}$   & ${\varepsilon=8}$    & ${\varepsilon=10}$    & ${\varepsilon=1}$            & ${\varepsilon=2}$           & ${\varepsilon=5}$   & ${\varepsilon=8}$    & ${\varepsilon=10}$      \\

 & $\text{G-Shuff}$ & 67.41 & 73.68 & 77.71 & 78.62 & 80.24 &77.86 & 79.55 & 80.08 & 80.17 & 80.19 \\
& $\text{G-BMF}$~(Setting~1)    & 75.28 & 79.11 & 80.57 & 80.82 & 80.88 &77.35& 77.95 & 78.09 & 78.14 & 78.17 \\
& $\text{G-BMF}$~(Setting~2)    & 68.12 & 74.51 & 80.94 & 82.76 &83.34& 77.39 & 81.59 & 83.62 & 84.05 & 84.14 \\
& $\text{GL-BMF}$~(Setting~1)   & 75.93 & 79.65 & 80.83 & 80.99 & 81.02 & 77.62 & 78.08 & 78.22 & 78.23 & 78.27  \\
& $\text{GL-BMF}$~(Setting~2)   & 68.95 & 75.22 & 81.58 & 83.21 & 83.68 & 77.81 & 81.86 & 83.79 & 84.15 & 84.26  \\
\bottomrule
\end{tabular}
}
\caption{Test Accuracy~($\%$) versus Different Datasets, Pre-training, and Privacy Budgets~(100 epochs)}
\label{tab1_1}
\end{table*}
\begin{table*}[!htbp]
\centering
{\small
\begin{tabular}{cc|cc|cc}
\toprule
\multirow{2}{*}{\textbf{Dataset}}    & \multirow{2}{*}{\textbf{Method}}          
&\multicolumn{2}{c|}{\textbf{Pre-ImageNet}} &\multicolumn{2}{c}{\textbf{Pre-Cifar100}} \\
 & & {\textbf{Runtime}~(s)} & {\textbf{Communication}~(MB)} & {\textbf{Runtime}~(s)} &{\textbf{Communication}~(MB)} \\
\midrule
\multirow{4}{*}{CIFAR-10}
&  $\text{G-Shuff}$ & 83.94 & 473.63 & 44.62 & 244.03 \\
&  $\text{G-BMF}$ & 82.60 & 453.6 & 43.92 & 233.77 \\
& $\text{GL-BMF}$ & 686.60 & 6220.87 & 468.19 & 3186.62 \\
 \midrule
\multirow{4}{*}{EMNIST}
 & $\text{G-Shuff}$ & 194.17 & 1228.56 & 108.34 & 632.52 \\
 & $\text{G-BMF}$ & 190.75 & 1178.64 & 106.65 & 606.96 \\
 & $\text{GL-BMF}$ & 725.44 & 6945.87 & 490.25 & 3559.79 \\
\bottomrule
\end{tabular}
}
\caption{Runtime and Communication cost of the MPC over WAN}
\label{tab:mpcwan}
\end{table*}
\begin{table*}[!htbp]
\centering
{\small
\begin{tabular}{c|cccc}
\toprule
\multirow{2}{*}{\textbf{Dataset}}          
&\multicolumn{3}{c}{\textbf{Method}} \\
 & {\text{GL-BMF}}  & {\text{GL-Sub}}  & {\text{ADMM}} & $\text{FedBCD}$  \\
 \midrule
\multirow{1}{*}{CIFAR-10} & 75.47 & 84.46 & 69.64 & 71.82  \\
\bottomrule
\end{tabular}
}
\caption{{Test Accuracy~(\%) versus Different Baselines~(100 epochs, $\epsilon = 4$, $\gamma = 2$)}}
\label{extend}
\end{table*}

\section{Multiparty Computation}
\label{app:mpc}

{\paragraph{Corruption} Correctness and privacy are guaranteed by assuming a certain type of adversary in the system. Depending on how we model the adversary there are different settings possible. The two big levers are the number of parties the adversary can control, and what the adversary can do with the corrupted parties. If the adversary can control only up to half of the parties ($t < K/2$) where $K$ is the number of the total parties, we are in an honest majority setting, whereas if it is higher than half ($t \geq K/2$), we are in dishonest majority. An adversary that only passively listens to the messages in the computation and tries to learn as much as possible from them, is a passive or semi-honest adversary. The adversary in this case will always send the correct message(s), and will not deviate from the expected behavior. If the adversary can behave in arbitrary ways, such as not sending a message, sending the wrong data, etc., it is referred to as an active or malicious adversary.

Active adversaries are much harder to deal with, and in some situations are considered overkill. For instance, when we are dealing with large companies performing a PPML task together, some have argued~\cite{matthew2025covert} that the reputational risk of actively deviating from the protocol is \emph{higher} than what could be gained by cheating. Passive protocols on the other hand are easier to design, and often serve as a first step towards building an actively secure protocol. This could be done by identifying the parts in the passive protocol that are vulnerable to active attacks, and proposing techniques to secure them.

Following the corruption threshold, in any MPC protocol, if more parties than the assumed threshold collude, security collapses and intermediate values can be exposed. Below this threshold, however, our MPC+DP protocol ensures that servers only see DP-protected outputs:~the clipped, aggregated, and noised gradients, providing an additional layer of privacy even if partial information is inadvertently revealed. In our MPC implementation, we deploy three servers tolerating one corrupted server, which balances efficiency and security; the protocol can be extended to more servers at higher computational cost. In general, security under a corruption-threshold assumption is all-or-nothing: no leakage occurs below the threshold, while full leakage is possible above it.

}
\paragraph{Secret-Sharing} We use the notation  to denote a value that is secret-shared between parties. In this work, we assume that values are secret-shared using a linear secret-sharing scheme, such as arithmetic secret-sharing, or replicated secret-sharing~\cite{CCS:AFLNO16}. If a value $x$ was secret-shared between three parties using replicated secret-sharing over a finite field $\F_p$, where $p$ is a prime, each party $P_i$ holds two secret-shares corresponding to $x$. For example, when , $P_0$ will hold $\ash{x}_1, \ash{x}_2$, $P_1$ holds $\ash{x}_2, \ash{x}_0$, and $P_2$ will hold $\ash{x}_0, \ash{x}_1$. 

The linearity of the secret-sharing scheme, such as the one mentioned earlier, allows for parties to compute additions of two secret-shared values locally. To compute a secret-sharing of $z = x + y$, each party simply adds the shares it holds corresponding to the individual values of $x, y$, resulting in $\ash{z}$. To compute a multiplication gate however, parties must communicate with each other. Protocols for efficient multiplication has received a lot of attention over the recent years, with \cite{CCS:AFLNO16,tetrad,USENIX:DalEscKel21,payman2018aby3,CCS:EGPS22} proposing efficient protocols for multiplication in the honest-majority setting, and works such as \cite{C:EGKRS20,C:RacSch22,C:CDESX18,payman2017secureml} to name a few in the dishonest-majority setting.


\paragraph{Pre-processing} MPC protocols sometimes involve using correlated randomness, such as Beaver triples~\cite{C:Beaver91b}, Oblivious Transfer (OT)~\cite{EPRINT:Rabin05}, or edaBits~\cite{C:EGKRS20} in our case. These are typically function-dependent but input-independent, which means that if the parties involved in the MPC know in advance what function they are going to compute, they can pre-process the correlated randomness required to carry out the MPC. Pre-processing takes a significant amount of time, especially for functions such as deep neural networks, so it makes a big difference in performance if one can take advantage of pre-processing. 


\paragraph{Mixed-mode Computation} Typically, computations in MPC are expressed as arithmetic ($\F_p$) or boolean circuits ($\F_2$), where $p$ is either a prime or $2^k$. Each of these computational domains has its advantages. For instance, performing multiplications is faster in the arithmetic domain, while an operation such as a secure comparison is more well suited for the boolean domain. The structure of neural networks is such that multiplications and non-linear operations such as comparisons and activation functions alternate. This makes it so that sticking to one of the two computational domains is suboptimal. One of the interesting techniques that has been proposed recently, in works such as \cite{demmler2015aby,payman2018aby3,payman2017secureml} is to have efficient ways to switch the secret-sharing domain from arithmetic to boolean and back, mid computation. These techniques allow for parties to take full advantage of the two domains.

{\subsection{Proof of Theorem~\ref{thm:input_privacy_main}}\label{sec:input_privacy_formal}

The proof, including a description of all building blocks and a simulator construction, is given below:

\paragraph{Building Blocks}
Most of our sub-protocols come from ABY3~\cite{payman2018aby3}, which is proven UC-secure against a static, semi-honest adversary corrupting at most one of three servers:
\begin{itemize}[leftmargin=*]
    \item $\PMult$ (multiplication over replicated secret shares): ABY3 uses the protocol from Araki et~al.~\cite{CCS:AFLNO16}, which is proven UC-secure in the $\mathcal{F}_{ABB}$-hybrid model.

    \item $\SHU$ (oblivious shuffle): built from local random permutations and re-sharing, both $\mathcal{F}_{ABB}$ operations.
    The construction follows Laur et~al.~\cite{laur2011shuffle} and is also used as a building block in the sorting protocol of Asharov et~al.~\cite{CCS:AHIKNPTT22} (Protocol~3.2); We further refer the reader to Theorem~4.6 in~\cite{CCS:AHIKNPTT22}, which provides a simulation-based proof of the full sorting functionality).

    \item $\PNM$ and $\PDiv$ (norm and division over fixed-point numbers): constructed by Aly and Smart~\cite{ACNS:AlySma19} in the $\mathcal{F}_{ABB}$-hybrid model using secure comparison, fixed-point division (from Catrina and Saxena~\cite{FC:CatSax10}), and bit decomposition, each independently proven secure under composition.

    \item $\PAct$ (activation functions): in our protocols the global model uses a softmax activation, which is computed from the exponentiation and division protocols of Aly and Smart~\cite{ACNS:AlySma19}, the same $\mathcal{F}_{ABB}$ operations used by $\PNM$ and $\PDiv$ above.

    \item $\PGS$ (secure Gaussian noise generation): we use the specific instantiation from \cite{keller2024secure}, which is built from arithmetic and boolean operations with secure domain conversions, which we use from ABY3~\cite{payman2018aby3}.

    \item $\FWD$ and $\BWD$ (forward and backward pass): sequential compositions of $\PMult$ and $\PAct$.
\end{itemize}
All subprotocols are instantiated in the MP-SPDZ~\cite{CCS:Keller20} framework.

\paragraph{No Leakage from Intermediate Values}
No secret-shared value is ever reconstructed between sub-protocol calls.
The output shares of one subprotocol (e.g., $\SHU$) are passed directly as input shares to the next (e.g., $\FWD$).
Since each server only ever holds its replicated shares, which are uniformly random and independent of the underlying secret~\cite{CCS:AFLNO16}, no information is leaked at the transitions between subprotocols.
The only reconstruction in the entire protocol occurs at the final step, when the DP-protected output is revealed.

Building on the above discussion, we first present the proof for protocol \PFirstA, which can be easily extended to the remaining protocols.

\begin{corollary}[Input privacy of $\PFirstA$]                         
$\PFirstA$ (\figref{pfirst1}) is secure against a static, semi-honest adversary corrupting at most one of   
the three servers, in the $\mathcal{F}_{ABB}$-hybrid model. 
\end{corollary}
\begin{proof}
Every operation in $\PFirstA$ such as secret sharing, shuffling, forward pass, backpropagation, gradient clipping, noise addition, and model update, is realized by $\mathcal{F}_{ABB}$ operations as established above.
By the UC composition theorem~\cite{FOCS:Canetti01}, the composed protocol is secure in the $\mathcal{F}_{ABB}$-hybrid model.

\paragraph{Simulator Construction}
We construct a simulator~$\Sim$ for the corrupted server~$P_c$ ($c\!\in\!\{1,2,3\}$).
Following the semi-honest simulation paradigm, $\Sim$ is given $P_c$'s \emph{input} (the collection of secret shares received from clients) and $P_c$'s \emph{output} (its share of the trained global model), and must produce a simulated \emph{view} consisting of $P_c$'s input, random tape, and all protocol messages received by $P_c$, that is computationally indistinguishable from $P_c$'s real view.

$\Sim$ then proceeds as follows:
\begin{enumerate}[leftmargin=*]
    \item $\Sim$ samples a uniformly random string~$r$ of appropriate length as $P_c$'s simulated random tape.

    \item For each $\mathcal{F}_{ABB}$ operation, $\Sim$ generates the messages $P_c$ would have received by running the operation's guaranteed simulator on $P_c$'s input and output shares for that operation.

    \item Local operations (additions, public-scalar multiplications) on shares are deterministic functions of the shares already present in the view, so $\Sim$ computes these consistently.
    By the linearity of replicated secret sharing~\cite{CCS:AFLNO16}, each party's shares remain uniformly distributed and independent of the underlying secret.
\end{enumerate}

\paragraph{Indistinguishability}
The simulated view is computationally indistinguishable from $P_c$'s real view because each $\mathcal{F}_{ABB}$ operation's simulator produces an indistinguishable transcript, and all local share operations are deterministic functions of identically distributed inputs.
\end{proof}

Typically, the same argument extends to the remaining protocols:

\begin{corollary}[Input privacy of $\PFirstB$]\label{cor:input_privacy_bmf}
$\PFirstB$ (\figref{pfirst2}) is secure against a static, semi-honest adversary corrupting at most one server, in the $\mathcal{F}_{ABB}$-hybrid model.
\end{corollary}
\begin{proof}
The proof for protocol $\PFirstB$ is identical to the proof for $\PFirstA$, with $\SHU$ absent (since $\PFirstB$ does not shuffle). The BandMF noise computation $\ash{\boldsymbol{\mathcal{N}}_{\text{Corr}}}=\boldsymbol{\Omega}^{-1}\ash{\boldsymbol{\mathcal{N}}_{\text{Tab}}}$ involves only a local linear operation on shares (as $\boldsymbol{\Omega}^{-1}$ is a public matrix) and requires no additional sub-protocol call.
\end{proof}

\begin{corollary}[Input privacy of $\PSecond$]\label{cor:input_privacy_gl}
$\PSecond$ (\figref{psecond}) is secure against a static, semi-honest adversary corrupting at most one server, in the $\mathcal{F}_{ABB}$-hybrid model.
\end{corollary}
\begin{proof}
The proof for $\PSecond$ extends Corollary~\ref{cor:input_privacy_bmf}: the additional $\PMult$ calls used to compute $\ash{\boldsymbol{g}_{\boldsymbol{\phi}^{L}_{i}, j}}$ via Eq.~\eqref{eq:multip} are $\mathcal{F}_{ABB}$ operations.
The reconstruction step, in which servers open $\Tilde{\boldsymbol{g}}_{\boldsymbol{\phi}^{L}_{i}}$ to client~$\cl{i}$, reveals only the DP-protected gradient and introduces no additional leakage.
\end{proof}

Based on the proofs for all above corollaries, we claim we have completed the proof for Theorem~\ref{thm:input_privacy_main}. 
}
\section{Differential Privacy}\label{app:dp}

\begin{definition}[Differential Privacy~\cite{dwork2014algorithmic}]\label{def:dp}
A randomized algorithm $\mathcal{A}$~ is said to be ($\epsilon$, $\delta$) differentially private if for any two neighboring input datasets $\mathcal{D}$ and $\mathcal{D}^{\prime}$, and for any subset $S$ of the output of $\mathcal{A}$, it always holds that
\begin{align}
\text{Pr}[\mathcal{A}(\mathcal{D}) \in S]\leq e^{\epsilon} \text{Pr}[\mathcal{A}(\mathcal{D}^\prime) \in S] + \delta ~~, \nonumber
\end{align}
for non-negative constants $\epsilon$ and $\delta$.
\end{definition}

\begin{definition}[Gaussian Mechanism~\cite{dwork2014algorithmic}]\label{def:gaussian}
Given a $d$-dimensional function $f$, with $\ell_2$ sensitivity $\Delta_2 f$ which is defined as 
\begin{align}
\Delta_2 f = \max\|f(\mathcal{D}) - f(\mathcal{D}^{\prime}) \|_2~~, \nonumber
\end{align}
for any two neighboring input datasets $\mathcal{D}$ and $\mathcal{D}^{\prime}$, for $\sigma \geq c\Delta_2 f/\epsilon $ and $c^2 > 2\ln(1.25/\delta)$, the Gaussian Mechanism with parameter $\sigma$ adds noise scaled to $\mathcal{N}(0, \sigma^2)$ to each of
the $d$ components of the output of $f$, is said to be ($\epsilon$, $\delta$) differentially private.
\end{definition}

\begin{definition}[Post-Processing~\cite{dwork2014algorithmic}]\label{def:post}
Given randomized algorithm $\mathcal{A}$ which is said to be ($\epsilon$, $\delta$) differentially private. Let $\mathcal{M}_R$ be an arbitrary randomized mapping, then $\mathcal{M}_R(\mathcal{A})$ is said to be ($\epsilon$, $\delta$) differentially private.  
\end{definition}

\begin{definition}[Sensitivity for the Matrix Factorization based Mechanism]\label{def:sens_MF}
The sensitivity of the matrix factorization mechanism Eq.~\eqref{eq:mf_eq} is defined as
\begin{align}
\text{sens}(\boldsymbol{C}_{\text{MF}}) = \sup_{\boldsymbol{G}~\sim~\boldsymbol{G}^{\prime}} \|\boldsymbol{C}_{\text{MF}}\boldsymbol{G} - \boldsymbol{C}_{\text{MF}}\boldsymbol{G}^{\prime} \|_{F}~~,
\end{align}
where $\boldsymbol{G}~\sim~\boldsymbol{G}^{\prime}$ denotes the \emph{neighborhood relation}, which indicates that the two data streams $\boldsymbol{G} $ and $\boldsymbol{G} ^{\prime}$, are differ in the contributions derived from exactly a single example. 
\end{definition}

\begin{definition}[Sensitivity for Decreasing Non-negative Toeplitz Matrices]\label{def:sens}
Suppose $\Phi = \text{LDToep}(\phi_0,\dots, \phi_{n-1}) \in \mathcal{R}^{n \times n}$ is a lower triangular Toeplitz matrix, with decreasing non-negative entries, as $\phi_0\geq \phi_1 \geq \dots \phi_{n-1} \geq 0$, then the sensitivity of $\Phi$ under $\text{b}$-min-separation is defined as
\begin{align}
\text{sens}_{\kappa,\text{b}}(\Phi) = \left\| \sum^{\kappa - 1}_{j=0} \Phi[:, 1 + j\text{b}]\right\|_{F}~, \nonumber
\end{align}
where $\Phi[:, 1 + j\text{b}]$ is the $1 + j\text{b}$-th column of $\Phi$
\end{definition}

\begin{definition}[Matrix Factorization Mechanism~\cite{denisov2022improved}]\label{def:maxtri_factorization}
Let $\boldsymbol{A}$ be a lower-triangular full-rank query matrix, and let $\boldsymbol{A}=\boldsymbol{B}_{\text{MF}}\boldsymbol{C}_{\text{MF}}$ be any factorization which has the following property that for any two neighboring streams of vectors $\boldsymbol{G} $ and $\boldsymbol{G} ^{\prime}$, we have 
\begin{align}
\|\boldsymbol{C}_{\text{MF}}\boldsymbol{G} - \boldsymbol{C}_{\text{MF}}\boldsymbol{G}^{\prime} \|_{F} \leq \zeta~~. \nonumber   
\end{align}
Let $\boldsymbol{Z}$ satisfied the distribution of $\mathcal{N}(0, \sigma^2\zeta^2)$ where $\sigma$ is large enough, so that the mechanism 
\begin{align}
\mathcal{M}_{\text{MF}}(\boldsymbol{G})=\boldsymbol{B}_{\text{MF}}(\boldsymbol{C}_{\text{MF}}\boldsymbol{G} + \boldsymbol{Z})~~, \nonumber
\end{align}
is said to be ($\epsilon$, $\delta$) differentially private in the nonadaptive continual release model. With the same parameters, mechanism $\mathcal{M}_{\text{MF}}(\cdot)$ is also said to satisfy the same DP guarantee even the input rows are chosen adaptively.
    
\end{definition}


{\subsection{Proof of Theorem~\ref{them1}}~\label{sec:proof}
\begin{proof}
To the protocol in~\figref{pfirst1}, 
at each step, as the $\ell_2$ sensitivity $\Delta_2 f$ of $\ash{\overline{\boldsymbol{g}}_{\theta}} $ is bounded by $\gamma$ through the clipping, the $(\epsilon, \delta)$ differential privacy guarantee of $\ash{\tilde{\boldsymbol{g}}_{\theta}}$ follows from the application of the Gaussian Mechanism~(\ref{def:gaussian}). Here let's assume each step satisfies $(\epsilon_1, \delta_1)$-DP.

With local models are frozen, and subsampling with the sampling ratio of $B/M$ is conducted through running MPC secure shuffling protocols at servers, the final privacy budget spent on global model training follows from the application of the the R\'enyi Differential Privacy (RDP) framework~\cite{Mironov2017renyi} for privacy accounting, which formalizes the moments accountant method introduced by~\cite{abadi2016privacy}. Furthermore, the RDP at order $\alpha$ for subsampled Gaussian mechanism is computed following~\cite{wang2019subsampled}. Hence, the RDP guarantee at order $\alpha^o$ composes linearly over $T_{num}$ steps, and we convert to $(\epsilon, \delta)$-DP through
\begin{align}
\epsilon = \min_{\alpha^o > 1} \{ \epsilon_{\alpha^o}^{(T_{num})} + \log(1/\delta)/(\alpha^o - 1) \}~~. \nonumber 
\end{align}
where $\epsilon_{\alpha^o}^{(T_{num})} = T_{num}\cdot\epsilon_{\alpha^o}^{(1)}$. Note here $(\epsilon_1, \delta_1)$-DP equivalently satisfies $(\alpha^o, \epsilon_{\alpha^o}^{(1)})$-RDP, and $\epsilon_{\alpha^o}^{(1)}$ is a function of the sampling ratio $B/M$ and other related mechanism parameters. For a detailed computation on $\epsilon_{\alpha^o}^{(1)}$, we refer the reader to~\cite{Mironov2017renyi} and skip the derivation here. We also refer the reader to~\cite{opacus} for implementation details. 

Therefore, the protocol in~\figref{pfirst1} is $(\epsilon, \delta)$ differentially private with respect to each private dataset $\mathcal{D}_i$, with $i \in [1, N]$.  

\end{proof}





\subsection{Proof of Theorem~\ref{them2}}~\label{sec:proof2} 
\begin{proof}
To the protocol in~\figref{pfirst2}, with local models are frozen and the $\text{b}$-min-separated participation that applied on all clients' datasets is pre-defined, the workload matrix $\boldsymbol{A}$ as well as the $p$-BSR matrix $\boldsymbol{\Omega}$ are therefore constructed.
Furthermore, the sensitivity of $\boldsymbol{\Omega}$, denoted~$\text{sens}_{\kappa,\text{b}}(\Omega)$, is computed and defined. 

At each step, due to the clipping, the $\ell_2$ sensitivity $\Delta_2 f$ of $\ash{\overline{\boldsymbol{g}}_{\theta}} $ is bounded by $\gamma$. Moreover, we follow the same setting in~\cite{Kalinin2024banded} to use FFT-based DP accounting~\cite{Koskela2020fft} to compute the appropriate $\sigma$ given target $(\epsilon, \delta)$. As a result, the $(\epsilon, \delta)$ DP guaranteed Gaussian matrix $\ash{\boldsymbol{\mathcal{N}}} \sim \mathcal{N}(\mathbf{0}, \sigma^2 I_{T_{num} \times n_{\boldsymbol{\theta}}})$ is generated~(under secured MPC) for the entire $T_{num}$ steps, and the correlated matrix 
\[\ash{\boldsymbol{\mathcal{N}}_{\text{Corr}}}= \boldsymbol{\Omega}^{-1}\ash{\boldsymbol{\mathcal{N}}_{\text{Tab}}}\] 
where $\ash{\boldsymbol{\mathcal{N}}_{\text{Tab}}} = \gamma \text{sens}_{\kappa,\text{b}}(\Omega)\ash{\boldsymbol{\mathcal{N}}}$ is obtained, which is still $(\epsilon, \delta)$ satisfied, guaranteed by post-processing~(\ref{def:post}). Hence, the training mechanism that applied on the \emph{stream} of $\ash{\overline{\boldsymbol{g}}_{\theta}} $ across all steps is exact the Matrix Factorization Mechanism~(\ref{def:maxtri_factorization}), which makes the \emph{stream} of $\ash{\Tilde{\boldsymbol{g}}_{\boldsymbol{\theta}}}$ across all steps be $(\epsilon, \delta)$ differentially private. We also refer the reader to~\cite{gdpl} for implementation details. 

Therefore, the protocol in~\figref{pfirst2} is $(\epsilon, \delta)$ differentially private with respect to each private dataset $\mathcal{D}_i$, with $i \in [1, N]$. 
\end{proof}

\subsection{Proof of Theorem~\ref{them3}}~\label{sec:proof3} 
\begin{proof}
To the protocol in~\figref{psecond}, at each step, the information that revealed by the servers is $\ash{\Tilde{\boldsymbol{g}}} = \frac{1}{B}(\ash{\overline{\boldsymbol{g}}} + \ash{\boldsymbol{\mathcal{N}}_{\Tilde{\boldsymbol{g}}} })$. As discussed in~\ref{sec:proof2}, with the sensitivity of $\boldsymbol{\Omega}$ is defined, as well as the $\ell_2$ sensitivity $\Delta_2 f$ of $\ash{\overline{\boldsymbol{g}}} $ is also bounded by $\gamma$, the \emph{stream} of $\ash{\Tilde{\boldsymbol{g}}} $ across all $T_{num}$ steps is $(\epsilon, \delta)$ differentially private under the Matrix Factorization mechanism~(\ref{def:maxtri_factorization}). 

As the output of the differentially private training mechanism technically consists of the full sequence of per-iteration privatized gradients, hence, as the intermediate result, the privatized gradient at each step $\ash{\Tilde{\boldsymbol{g}}} $ is also DP protected~\cite{Choquette-Choo2023amplified}. Therefore, at each step, at the global model part, $\ash{\Tilde{\boldsymbol{g}}_{\boldsymbol{\theta}}} $, is naturally differentially private as $\ash{\Tilde{\boldsymbol{g}}_{\boldsymbol{\theta}}} $ is a sub-vector of $\ash{\Tilde{\boldsymbol{g}}}$. At each local model part, as a sub-vector of $\Tilde{\boldsymbol{g}}$, $\Tilde{\boldsymbol{g}}_{\boldsymbol{\phi}^{L}_{i}}$ is also $(\epsilon, \delta)$ differentially private. 

Furthermore, all the information with respect to $\mathcal{D}_j, j \ne i$ that is derived by each client $E_i$ from $\Tilde{\boldsymbol{g}}_{\boldsymbol{\phi}^{L}_{i}}$ is also $(\epsilon, \delta)$ differentially private, guaranteed by the post-processing~(\ref{def:post}) of differential privacy:~In~\figref{psecond}, at each client, here $\mathcal{M}_R(\cdot)$, as defined in~\ref{def:post}, functions to extract estimates of the per-sample gradients from $\Tilde{\boldsymbol{g}}_{\boldsymbol{\phi}^{L}_{i}}$. 

Therefore, the protocol in~\figref{psecond} is $(\epsilon, \delta)$ differentially private with respect to each private dataset $\mathcal{D}_i$, with $i \in [1, N]$. 
\end{proof}
}
\section{Estimation Error Bound Analysis}\label{app:error}

Considering to solve the following linear system
\begin{align}
\Tilde{\boldsymbol{g}}_{\boldsymbol{\phi}^{L}_{i}} = \frac{1}{B}({\interout}_{{\boldsymbol{\phi}}^{L}_i} \cdot  \widehat{\boldsymbol{g}}_{\textbf{H}^{(i)}_b} + \boldsymbol{\mathcal{N}}_i)~~, \quad \text{with} \quad \|\widehat{\boldsymbol{g}}_{\textbf{H}^{(i)}_b}\|_{\ell_2} \le \gamma~~\nonumber 
\end{align}
where $\boldsymbol{\mathcal{N}}_i \sim \mathcal{N}(0, \sigma_t^2 I)$ and $\sigma_t = \sigma \gamma \text{sens}_{\kappa,\text{b}}(\boldsymbol{\Omega})\|\boldsymbol{\Omega}^{-1}[t,:]\|^2_{\ell_2}$. First, let's define the effective noise variance after scaling as:
\begin{align}
\sigma_{t(B)}^2 = \frac{\sigma_t^2}{B^2}~~. \nonumber
\end{align} 

As the RR estimator with ridge parameter $\lambda = \sigma_{t(B)}^2$ is defined by
\begin{align}
\widehat{\boldsymbol{g}}^{\text{ridge}}_{\textbf{H}^{(i)}_b} = ({\interout}_{{\boldsymbol{\phi}}^{L}_i}^T {\interout}_{{\boldsymbol{\phi}}^{L}_i} + \sigma_{t(B)}^2 I)^{-1} {\interout}_{{\boldsymbol{\phi}}^{L}_i}^T \Tilde{\boldsymbol{g}}_{\boldsymbol{\phi}^{L}_{i}}~~, \nonumber 
\end{align}
the estimation error can be decomposed as
\begin{align}
\widehat{\boldsymbol{g}}^{\text{ridge}}_{\textbf{H}^{(i)}_b} - \widehat{\boldsymbol{g}}_{\textbf{H}^{(i)}_b}
&= ({\interout}_{{\boldsymbol{\phi}}^{L}_i}^T {\interout}_{{\boldsymbol{\phi}}^{L}_i} + \sigma_{t(B)}^2 I)^{-1} {\interout}_{{\boldsymbol{\phi}}^{L}_i}^T \Tilde{\boldsymbol{g}}_{\boldsymbol{\phi}^{L}_{i}} - \widehat{\boldsymbol{g}}_{\textbf{H}^{(i)}_b} \nonumber \\
&= \underbrace{- \sigma_{t(B)}^2 ({\interout}_{{\boldsymbol{\phi}}^{L}_i}^T {\interout}_{{\boldsymbol{\phi}}^{L}_i} + \sigma_B^2 I)^{-1} \widehat{\boldsymbol{g}}_{\textbf{H}^{(i)}_b}}_{\text{bias term}} + \nonumber \\
& \underbrace{({\interout}_{{\boldsymbol{\phi}}^{L}_i}^T {\interout}_{{\boldsymbol{\phi}}^{L}_i} + \sigma_{t(B)}^2 I)^{-1} {\interout}_{{\boldsymbol{\phi}}^{L}_i}^T \frac{\mathcal{N}_i}{B}}_{\text{variance term}}~~. \nonumber 
\end{align}

Furthermore, taking expectation over the noise, which results in
\begin{align}
\mathbb{E}_{\text{Ridge}} = \mathbb{E}\|\widehat{\boldsymbol{g}}^{\text{ridge}}_{\textbf{H}^{(i)}_b} - \widehat{\boldsymbol{g}}_{\textbf{H}^{(i)}_b}\|_{\ell_2}^2
= \sigma_{t(B)}^4 \|({\interout}_{{\boldsymbol{\phi}}^{L}_i}^T {\interout}_{{\boldsymbol{\phi}}^{L}_i} + \sigma_B^2 I)^{-1} \widehat{\boldsymbol{g}}_{\textbf{H}^{(i)}_b}\|_{\ell_2}^2
+ \nonumber \\
\sigma_{t(B)}^2 \, \mathrm{tr}\Big(({\interout}_{{\boldsymbol{\phi}}^{L}_i}^T {\interout}_{{\boldsymbol{\phi}}^{L}_i} + \sigma_{t(B)}^2 I)^{-1} {\interout}_{{\boldsymbol{\phi}}^{L}_i}^T  {\interout}_{{\boldsymbol{\phi}}^{L}_i} ({\interout}_{{\boldsymbol{\phi}}^{L}_i}^T {\interout}_{{\boldsymbol{\phi}}^{L}_i} + \sigma_{t(B)}^2 I)^{-1}\Big)~~, \nonumber 
\end{align}
where $\mathrm{tr}(\cdot)$ denotes the trace operation on matrix. Let $\mu_{\min}$ and $\mu_{\max}$ denote the smallest and largest eigenvalues of ${\interout}_{{\boldsymbol{\phi}}^{L}_i}^T {\interout}_{{\boldsymbol{\phi}}^{L}_i}$, respectively. By using $\|\widehat{\boldsymbol{g}}_{\textbf{H}^{(i)}_b}\|_{\ell_2} \le \gamma$ and the spectral norm bound on the variance term, we could bound the two terms as follows:

\paragraph{Bias term} As we have
\[
\|({\interout}_{{\boldsymbol{\phi}}^{L}_i}^T {\interout}_{{\boldsymbol{\phi}}^{L}_i} + \sigma_{t(B)}^2 I)^{-1}\|_{\ell_2} = 1/(\mu_{\min}+\sigma_{t(B)}^2)~~,
\]
by using the operator norm inequality, we could obtain
\begin{align}
&\sigma_{t(B)}^4 \|({\interout}_{{\boldsymbol{\phi}}^{L}_i}^T {\interout}_{{\boldsymbol{\phi}}^{L}_i} + \sigma_B^2 I)^{-1} \widehat{\boldsymbol{g}}_{\textbf{H}^{(i)}_b}\|_{\ell_2}^2
\nonumber \\
& \le \sigma_{t(B)}^4 \|({\interout}_{{\boldsymbol{\phi}}^{L}_i}^T {\interout}_{{\boldsymbol{\phi}}^{L}_i} + \sigma_B^2 I)^{-1}\|_{\ell_2}^2 \|\widehat{\boldsymbol{g}}_{\textbf{H}^{(i)}_b}\|_{\ell_2}^2
\nonumber \\ 
&\le \frac{\sigma_{t(B)}^4 \gamma^2}{(\mu_{\min} + \sigma_{t(B)}^2)^2}~~. \nonumber
\end{align}

\paragraph{Variance term}
Diagonalize \({\interout}_{{\boldsymbol{\phi}}^{L}_i}^T {\interout}_{{\boldsymbol{\phi}}^{L}_i}\) with eigenvalues \(\mu_i\) where $i= 1,\dots, d_{\interout^{(i)}} B$. The variance term equals
\begin{align}
 &\sigma_{t(B)}^2 \, \mathrm{tr}\Big(({\interout}_{{\boldsymbol{\phi}}^{L}_i}^T {\interout}_{{\boldsymbol{\phi}}^{L}_i} + \sigma_{t(B)}^2 I)^{-1} {\interout}_{{\boldsymbol{\phi}}^{L}_i}^T  {\interout}_{{\boldsymbol{\phi}}^{L}_i} ({\interout}_{{\boldsymbol{\phi}}^{L}_i}^T {\interout}_{{\boldsymbol{\phi}}^{L}_i} + \sigma_{t(B)}^2 I)^{-1}\Big) \nonumber \\
&  = \sigma_{t(B)}^2 \mathrm{tr}\!\big((\Lambda+\sigma_{t(B)}^2 I)^{-1}\Lambda(\Lambda+\sigma_{t(B)}^2 I)^{-1}\big) \nonumber \\
& = \sigma_{t(B)}^2\sum_{i=1}^{d_{\interout^{(i)}}B} \frac{\mu_i}{(\mu_i + \sigma_{t(B)}^2)^2}~~,  \nonumber 
\end{align}
where $\Lambda = \text{diag}(\mu_i)$ for $i=1,\dots,d_{\interout^{(i)}}B$. A relax upper bound is obtained by upper-bounding each fraction by \(\mu_{\max}/(\mu_{\min}+\sigma_{t(B)}^2)^2\) and summing, which results in
\[
\sigma_{t(B)}^2 \sum_{i=1}^{d_{\interout^{(i)}}B} \frac{\mu_i}{(\mu_i + \sigma_{t(B)}^2)^2}
\le \sigma_{t(B)}^2 \cdot d_{\interout^{(i)}}B \cdot \frac{\mu_{\max}}{(\mu_{\min} + \sigma_{t(B)}^2)^2}~~.
\]

Eventually, we could combine the two above bounds together, which yields the fully bounded error, as
\begin{align}
&\mathbb{E}_{\text{Ridge}}
\le
\frac{\sigma_{t(B)}^4 \gamma^2}{(\mu_{\min} + \sigma_{t(B)}^2)^2}
+
\frac{\sigma_{t(B)}^2 d_{\interout^{(i)}}B \, \mu_{\max}}{(\mu_{\min} + \sigma_{t(B)}^2)^2} \nonumber \\
& = \frac{\sigma_{t(B)}^4 \gamma^2 + \sigma_{t(B)}^2 d_{\interout^{(i)}}B\, \mu_{\max}}{(\mu_{\min} + \sigma_{t(B)}^2)^2} \nonumber \\
& = \frac{\sigma_{t}^4 \gamma^2 + \sigma_{t}^2 d_{\interout^{(i)}}B^3\, \mu_{\max}}{(B^2\mu_{\min} + \sigma_{t}^2)^2}~~.
\end{align}

As ${\interout}_{{\boldsymbol{\phi}}^{L}_i}$ is determined by the clients' data and local model parameters, the client can pre-compute the estimation error bound, which can serve as a “confidence measure” to monitor updates over iterations and track the reliability of the DP-noised gradients. Although the bound cannot be used to adjust batch size or DP noise (to preserve privacy), observing its trend provides valuable diagnostic information about the solver’s stability and the confidence of the updates, which in turn can inform the future choice of fixed parameters or guide offline experiments.

However, the above proposal is currently conceptual, and we leave it as a future direction to explore more systematic formulations, theoretical guarantees, and practical implementations that bridge the gap between DP-noised linear estimation and potential DP-safe adaptive optimization strategies.



\end{document}